\documentclass[a4paper, USenglish]{lipics-v2021}

\usepackage{xcolor}
\usepackage{graphicx}
\usepackage{array}
\usepackage{hyperref}

\usepackage{multirow}

\usepackage{caption}
\usepackage{subcaption}
\usepackage{mathtools}

\usepackage{varwidth}
\usepackage{algorithm}
\usepackage[noend]{algpseudocode}
\makeatletter
\renewcommand{\ALG@beginalgorithmic}{\ttfamily\footnotesize}
\makeatother

\newcommand{\N}{\mathbb{N}}
\newcommand{\set}[1]{\{#1\}}
\newcommand{\ps}{\mathcal{P}} 
\newcommand{\emptyseq}{\emptyset} 
\newcommand{\tuple}[1]{\langle#1\rangle} 

\newcommand{\bfit}[1]{\emph{{#1}}} 

\newcommand{\rsstyle}[1]{\underline{#1}} 
\newcommand{\scstyle}[1]{\colorbox{lightgray}{$#1$}} 
\newcommand{\untilstyle}[1]{\textcolor{cyan}{#1}} 

\newcommand{\atomicbegin}{{\bf \textcolor{teal}{atomic\big\{}}}
\newcommand{\atomicend}{{\bf \textcolor{teal}{\big\}}}}
\newcommand{\code}[2]{#1:#2}

\algblockdefx[upon]{When}{EndWhen}%
[3][Unknown]{\textbf{when for every} #3 \textbf{received a quorum of} #1(#2)}%
{\textbf{end}}

\algblockdefx[upon]{WhenReceive}{EndWhenReceive}%
[1][Unknown]{\textbf{when received} #1}%
{\textbf{end}}

\algblockdefx[upon]{Upon}{EndUpon}%
[1][Unknown]{\textbf{upon} #1}%
{\textbf{end upon}}

\makeatletter
\ifthenelse{\equal{\ALG@noend}{t}}%
{\algtext*{EndWhen}}
{}%
\makeatother

\makeatletter
\ifthenelse{\equal{\ALG@noend}{t}}%
{\algtext*{EndWhenReceive}}
{}%
\makeatother

\makeatletter
\ifthenelse{\equal{\ALG@noend}{t}}%
{\algtext*{EndUpon}}
{}%
\makeatother

\newcommand{\send}{{\bf send\;}}
\newcommand{\sendto}{{\bf \;to\;}}
\newcommand{\wait}{{\bf wait\;}}
\newcommand{\until}{{\bf until\;}}
\newcommand{\receive}{{\bf receive\;}}
\newcommand{\from}{{\bf \;from\;}}

\newcommand{\wtalg}{\textsc{WiredTiger}}
\newcommand{\rsalg}{\textsc{ReplicaSet}}

\newcommand{\scalg}{\textsc{ShardedCluster}}

\newcommand{\rollback}{{\sf rollback}}
\newcommand{\ok}{{\sf ok}}
\newcommand{\rc}{\emph{readConcern}\;}
\newcommand{\wc}{\emph{writeConcern}\;}

\newcommand{\majority}{\emph{majority}}
\newcommand{\snapshotrc}{\emph{snapshot}}
\newcommand{\noop}{{\sf noop}}
\newcommand{\T}{\mathbb{T}}
\newcommand{\Key}{{\sf Key}}
\newcommand{\Val}{{\sf Val}}
\newcommand{\wtdep}{\textsc{wt}}
\newcommand{\rsdep}{\textsc{rs}}
\newcommand{\scdep}{\textsc{sc}}
\newcommand{\WTSID}{{\sf WT\_SID}}
\newcommand{\wtsid}{{\sf wt\_sid}}
\newcommand{\wtsidvar}{\mathit{wt\_sid}}
\newcommand{\txn}{{\sf txn}}
\newcommand{\txnvar}{\mathit{txn}}
\newcommand{\wtsession}{{\sf wt\_session}}

\newcommand{\wttsnone}{\bot_{{\sf ts}}}
\newcommand{\wttidnone}{\bot_{{\sf tid}}}

\newcommand{\deployments}{{\sf Dep}}
\newcommand{\deployment}{\mathit{d}}
\newcommand{\id}{{\sf id}}
\newcommand{\TID}{{\sf TID}}
\newcommand{\tid}{{\sf tid}}
\newcommand{\tidvar}{\mathit{tid}}

\newcommand{\wttxnglobal}{{\sf wt\_global}}
\newcommand{\currenttid}{{\sf tid}}

\newcommand{\WTTXN}{{\sf WT\_TXN}}
\newcommand{\WTTXNUPDATE}{\sf WT\_TXN\_UPDATE}

\newcommand{\upperlimit}{{\sf limit}}
\newcommand{\concur}{{\sf concur}}
\newcommand{\mods}{{\sf mods}}
\newcommand{\store}{{\sf store}}

\newcommand{\kvar}{\mathit{k}}
\newcommand{\keyvar}{\mathit{key}}
\newcommand{\val}{{\sf val}}
\newcommand{\vvar}{\mathit{v}}
\newcommand{\valvar}{\mathit{val}}

\newcommand{\wtstart}{\textsc{wt\_start}}
\newcommand{\wtread}{\textsc{wt\_read}}
\newcommand{\wtupdate}{\textsc{wt\_update}}
\newcommand{\wtcommit}{\textsc{wt\_commit}}
\newcommand{\wtrollback}{\textsc{wt\_rollback}}
\newcommand{\txnvis}{\textsc{visible}}
\newcommand{\allcommitted}{\textsc{all\_committed}}
\newcommand{\wtsetcommitts}{\textsc{wt\_set\_commit\_ts}}

\newcommand{\RSTXN}{{\sf RS\_TXN}}
\newcommand{\RSTXNUPDATE}{\sf RS\_TXN\_UPDATE}
\newcommand{\RSSID}{{\sf RS\_SID}}

\newcommand{\rssidvar}{\mathit{rs\_sid}}

\newcommand{\OPLOG}{\textsc{oplog}}
\newcommand{\TS}{{\sf TS}}
\newcommand{\ts}{{\sf ts}}
\newcommand{\tsvar}{\mathit{ts}}
\newcommand{\readts}{{\sf read\_ts}}
\newcommand{\readtsvar}{\mathit{read\_ts}}
\newcommand{\maxcommitts}{{\sf max\_commit\_ts}}

\newcommand{\committs}{{\sf commit\_ts}}
\newcommand{\committsvar}{\mathit{commit\_ts}}
\newcommand{\pinnedcommittedts}{{\sf commit\_ts}}
\newcommand{\pinnedcommittedtsvar}{\mathit{commit\_ts}}
\newcommand{\ops}{{\sf ops}}
\newcommand{\opsvar}{\mathit{ops}}
\newcommand{\oplog}{{\sf oplog}}
\newcommand{\oplogvar}{\mathit{oplog}}

\newcommand{\rswtconns}{{\sf rs\_wt}}

\newcommand{\txnops}{{\sf txn\_mods}}

\newcommand{\rsstart}{\textsc{rs\_start}}
\newcommand{\rsread}{\textsc{rs\_read}}
\newcommand{\rsupdate}{\textsc{rs\_update}}
\newcommand{\rscommit}{\textsc{rs\_commit}}
\newcommand{\rsrollback}{\textsc{rs\_rollback}}

\newcommand{\openwtsession}{\textsc{open\_wt\_session}}

\newcommand{\ct}{\sf ct} 
\newcommand{\ctvar}{\mathit{ct}}

\newcommand{\tick}{\textsc{tick}}

\newcommand{\replicate}{\textsc{replicate}}

\newcommand{\pulloplog}{\textsc{pull\_oplog}}
\newcommand{\pushoplog}{\textsc{push\_oplog}}

\newcommand{\Primary}{\mathit{P}}

\newcommand{\primaryvar}{\mathit{p}}

\newcommand{\secondaryvar}{\mathit{s}}
\newcommand{\replicateack}{\textsc{replicate\_ack}}

\newcommand{\lastpulled}{{\sf last\_pulled}}
\newcommand{\lastpulledacks}{{\sf last\_pulled\_ack}}
\newcommand{\lastpulledvar}{\mathit{last\_pulled}}
\newcommand{\lastmajoritycommitted}{{\sf last\_majority\_committed}}

\newcommand{\SCTXN}{{\sf SC\_TXN}}
\newcommand{\SCTXNRO}{{\sf SC\_RO\_TXN}}
\newcommand{\SCTXNUPDATE}{{\sf SC\_UPDATE\_TXN}}
\newcommand{\ShardID}{{\sf SHARD\_ID}}

\newcommand{\statusvar}{\mathit{status}}
\newcommand{\phase}{{\sf phase}}
\newcommand{\phasevar}{\mathit{phase}}

\newcommand{\scrsconns}{{\sf sc\_rs}}

\newcommand{\shards}{{\sf shards}}

\newcommand{\primaryof}{\textsc{primaries\_of}}

\newcommand{\SCSID}{{\sf SC\_SID}}
\newcommand{\scsidvar}{\mathit{sc\_sid}}

\newcommand{\twopc}{\textsc{2pc}}
\newcommand{\twopcack}{\textsc{2pc\_ack}}

\newcommand{\mongosvar}{\mathit{m}}

\newcommand{\scread}{\textsc{sc\_read}}
\newcommand{\waitforreadconcern}{\textsc{wait\_for\_read\_concern}}
\newcommand{\scupdate}{\textsc{sc\_update}}
\newcommand{\scstart}{\textsc{sc\_start}}
\newcommand{\prepare}{\textsc{prepare}}
\newcommand{\prepareack}{\textsc{prepare\_ack}}
\newcommand{\commit}{\textsc{commit}}
\newcommand{\abort}{\textsc{abort}}
\newcommand{\decack}{\textsc{dec\_ack}}

\newcommand{\wtprepare}{\textsc{wt\_prepare}}
\newcommand{\wtprepareinprogress}{\textsc{prepared}}
\newcommand{\wtprepareresolved}{\textsc{committed}}
\newcommand{\wtsetreadts}{\textsc{wt\_set\_read\_ts}}

\newcommand{\preparets}{{\sf prepare\_ts}}
\newcommand{\preparetsvar}{\mathit{prepare\_ts}}

\newcommand{\wtcommitprepare}{\textsc{wt\_commit\_prepare}}
\newcommand{\wtcommitpreparets}{\textsc{wt\_commit\_prepare\_ts}}

\newcommand{\caseclockskew}{\case-\textsc{Clock-Skew}}
\newcommand{\casependingcommitread}{\case-\textsc{Pending-Commit-Read}}
\newcommand{\casependingcommitupdate}{\case-\textsc{Pending-Commit-Update}}
\newcommand{\caseholes}{\case-\textsc{Holes}}

\newcommand{\lclock}{{\sf lc}}
\newcommand{\E}{E}
\newcommand{\evar}{\mathit{e}}
\newcommand{\fvar}{\mathit{f}}
\newcommand{\Event}{{\sf Event}}

\newcommand{\HEvent}{{\sf HEvent}}
\newcommand{\op}{{\sf op}}
\newcommand{\Op}{{\sf Op}}
\newcommand{\readevent}{{\sf read}}

\newcommand{\writeevent}{{\sf write}}
\newcommand{\Write}{{\sf Write}\;}
\newcommand{\WriteTx}{{\sf WriteTx}}
\newcommand{\starttime}{\textsl{start}}
\newcommand{\starttimewt}{\textsl{start}_\textsl{wt}}

\newcommand{\committime}{\textsl{commit}}
\newcommand{\committimewt}{\textsl{commit}_\textsl{wt}}

\newcommand{\h}{\mathcal{H}}
\renewcommand{\ae}{\mathcal{A}}
\newcommand{\axiom}{\Phi}
\newcommand{\rel}[1]{\xrightarrow{#1}}

\newcommand{\comp}{\;;}
\newcommand{\po}{{\sf po}}
\newcommand{\sorel}{\textsc{so}}
\newcommand{\sowt}{\textsc{so}_{\textsc{wt}}}
\newcommand{\sors}{\textsc{so}_{\textsc{rs}}}
\newcommand{\sosc}{\textsc{so}_{\textsc{sc}}}
\newcommand{\rb}{\textsc{rb}}
\newcommand{\rbwt}{\textsc{rb}_{\textsc{wt}}}
\newcommand{\rbrs}{\textsc{rb}_{\textsc{rs}}}

\newcommand{\cb}{\textsc{cb}}
\newcommand{\cbwt}{\textsc{cb}_\textsc{wt}}
\newcommand{\cbrs}{\textsc{cb}_\textsc{rs}}

\newcommand{\vis}{\textsc{vis}}
\newcommand{\viswt}{\textsc{vis}_{\textsc{wt}}}
\newcommand{\visrs}{\textsc{vis}_{\textsc{rs}}}
\newcommand{\vissc}{\textsc{vis}_{\textsc{sc}}}
\newcommand{\ar}{\textsc{ar}}
\newcommand{\arwt}{\textsc{ar}_{\textsc{wt}}}
\newcommand{\arrs}{\textsc{ar}_{\textsc{rs}}}
\newcommand{\arsc}{\textsc{ar}_{\textsc{sc}}}
\newcommand{\hist}{\text{Hist}}
\newcommand{\intaxiom}{\textsc{Int}}
\newcommand{\extaxiom}{\textsc{Ext}}
\newcommand{\sessionaxiom}{\textsc{Session}}
\newcommand{\prefixaxiom}{\textsc{Prefix}}
\newcommand{\rbaxiom}{\textsc{ReturnBefore}}
\newcommand{\cbaxiom}{\textsc{CommitBefore}}
\newcommand{\conflict}{\bowtie}
\newcommand{\noconflictaxiom}{\textsc{NoConflict}}
\newcommand{\inrbaxiom}{\textsc{InReturnBefore}}
\newcommand{\realtimesnapshotaxiom}{\textsc{RealTimeSnapshot}}

\newcommand{\si}{\textsc{SI}}
\newcommand{\ansisi}{\textsc{ANSI-SI}}
\newcommand{\rtsi}{\textsc{RealtimeSI}}
\newcommand{\gsi}{\textsc{GSI}}
\newcommand{\parallelsi}{\textsc{PSI}}

\newcommand{\strongsi}{\textsc{StrongSI}}
\newcommand{\strongsessionsi}{\textsc{StrongSessionSI}}
\newcommand{\sessionsi}{\textsc{SessionSI}}
\newcommand{\nmsi}{\textsc{NMSI}}

\newcommand{\timepoint}{\tau}
\newcommand{\case}{\textsc{Case}}
\newcommand{\casei}{\case\; I}
\newcommand{\caseii}{\case\; II}
\newcommand{\wtvis}{\textsl{wt\_vis}}
\newcommand{\rsvis}{\textsl{rs\_vis}}
\newcommand{\scvis}{\textsl{sc\_vis}}

\newcommand{\param}[1]{\textbf{\##1}}

\algblock[TryCatchFinally]{try}{endtry}
\algcblock[TryCatchFinally]{TryCatchFinally}{finally}{endtry}
\algcblockdefx[TryCatchFinally]{TryCatchFinally}{catch}{endtry}
	[1]{\textbf{catch} #1}
    {\textbf{end try}}

\newcommand{\cyan}[1]{\textcolor{cyan}{#1}}

\title{Verifying Transactional Consistency of MongoDB}

\author{Hongrong Ouyang}
    {State Key Laboratory for Novel Software Technology, Nanjing, China 
    \and Nanjing University, China 
    \and \url{https://tsunaou.github.io/}}
    {mf20330056@smail.nju.edu.cn}
    {}{}
\author{Hengfeng Wei~\footnote{Corresponding author.}}
    {State Key Laboratory for Novel Software Technology, Nanjing, China 
    \and Nanjing University, China
    \and Software Institute at Nanjing University, China
    \and \url{https://hengxin.github.io/}}
    {hfwei@nju.edu.cn}
    {}{}
\author{Yu Huang}
    {State Key Laboratory for Novel Software Technology, Nanjing, China 
    \and Nanjing University, China 
    \and \url{https://cs.nju.edu.cn/yuhuang/}}
    {yuhuang@nju.edu.cn}
    {}{}
\author{Haixiang Li}
    {TDSQL Team of Technology and Engineering Group of Tencent, Tencent Inc., Shenzhen, China}
    {blueseali@tencent.com}
    {}{}
\author{Anqun Pan}
    {TDSQL Team of Technology and Engineering Group of Tencent, Tencent Inc., Shenzhen, China}
    {aaronpan@tencent.com}
    {}{}

\Copyright{Jane Open Access and Joan R. Public}

\ccsdesc[100]{}

\keywords{MongoDB, Snapshot isolation, Verification, Jepsen testing}
\category{Regular Paper}

\EventEditors{John Q. Open and Joan R. Access}
\EventNoEds{2}
\EventLongTitle{42nd Conference on Very Important Topics (CVIT 2016)}
\EventShortTitle{CVIT 2016}
\EventAcronym{CVIT}
\EventYear{2016}
\EventDate{December 24--27, 2016}
\EventLocation{Little Whinging, United Kingdom}
\EventLogo{}
\SeriesVolume{42}
\ArticleNo{23}
\begin{document}
\maketitle


\begin{abstract}
  MongoDB is a popular general-purpose, document-oriented,
  distributed NoSQL database.
  It supports transactions in three different deployments:
  single-document transactions utilizing the WiredTiger storage engine in a standalone node,
  multi-document transactions in a replica set
  which consists of a primary node and several secondary nodes,
  and distributed transactions in a sharded cluster
  which is a group of multiple replica sets, among which data is sharded.
  A natural and fundamental question about MongoDB transactions is:
  \emph{What transactional consistency guarantee do MongoDB transactions
  in each deployment provide?}
  However, it lacks both concise pseudocode of MongoDB transactions
  in each deployment and formal specification of the consistency guarantees
  which MongoDB claimed to provide.
  In this work, we formally specify and verify
  the transactional consistency protocols of MongoDB.
  Specifically, we provide a concise pseudocode for the transactional
  consistency protocols in each MongoDB deployment,
  namely \wtalg, \rsalg, and \scalg{},
  based on the official documents and source code.
  We then prove that \wtalg, \rsalg, and \scalg{}
  satisfy different variants of snapshot isolation,
  namely \strongsi, \rtsi, and \sessionsi, respectively.
  We also propose and evaluate efficient white-box checking algorithms
  for MongoDB transaction protocols against their consistency guarantees,
  effectively circumventing the \textsf{NP-hard} obstacle in theory.
\end{abstract}


\section{Introduction} \label{section:intro}

MongoDB is a popular general-purpose, document-oriented,
distributed NoSQL database~\footnote{MongoDB. \url{https://www.mongodb.com/}}.
A MongoDB database consists of a set of collections,
a collection is a set of documents,
and a document is an ordered set of
keys with associated values~\cite{VLDB2019:TunableConsistency}~\footnote{
  Roughly speaking, a document is an analog to a row in a relational database,
  and a collection is to a table.
}.
MongoDB achieves scalability by partitioning data into shards and fault-tolerance
by replicating each shard across a set of nodes~\cite{VLDB2019:TunableConsistency}.

\begin{figure}[t]
  \centering
  \includegraphics[width = 0.60\textwidth]{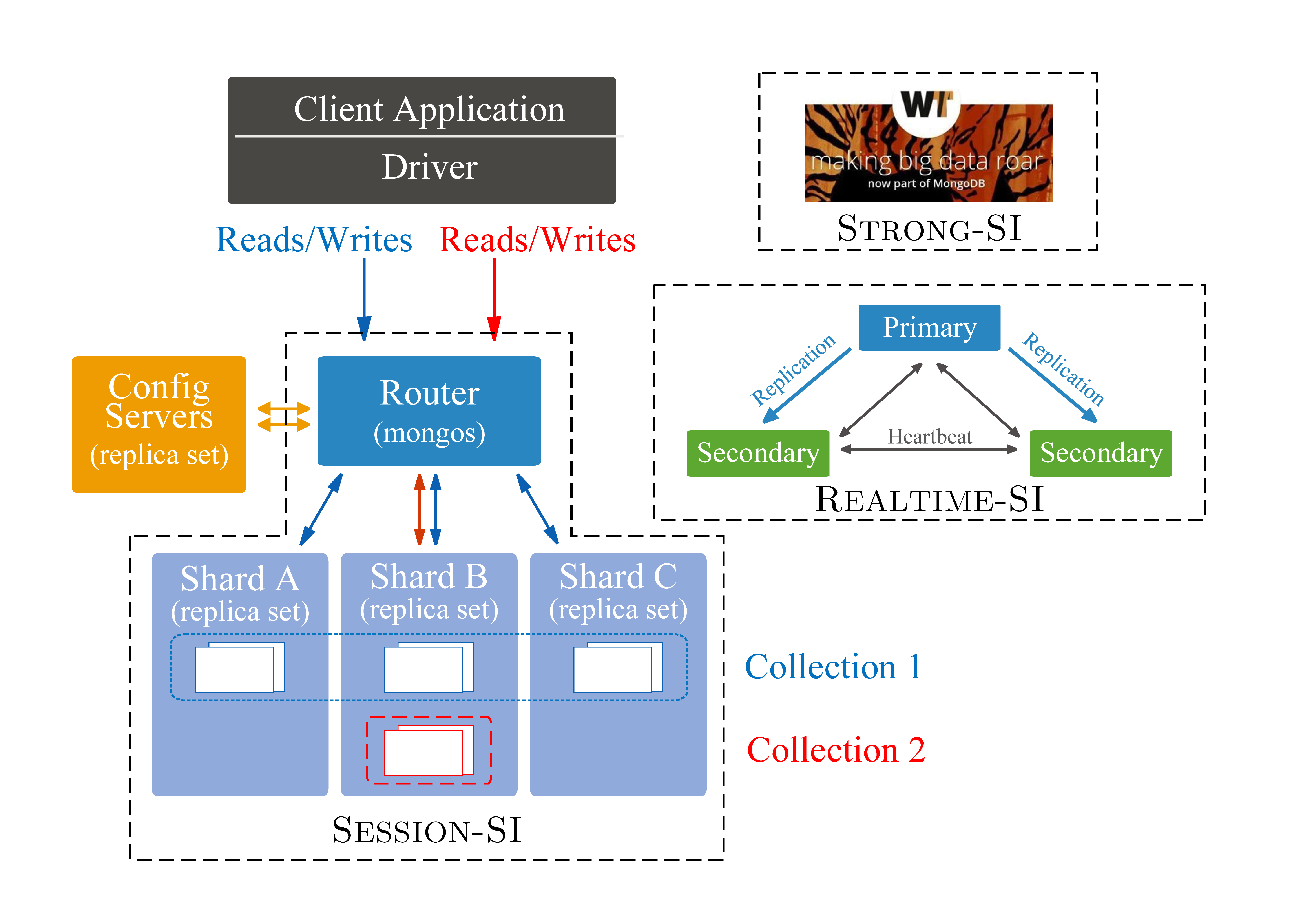}
  \caption{The MongoDB deployment.}
  \label{fig:mongodb-sharded-cluster}
\end{figure}

MongoDB deployment is a sharded cluster, replica set,
or standalone~\cite{SIGMOD2019:MongoDB-CC};
see Figure~\ref{fig:mongodb-sharded-cluster}.
A standalone is a storage node that represents a single instance of a data store.
A replica set consists of a primary node and several secondary nodes.
A sharded cluster is a group of multiple replica sets,
among which data is sharded.

MongoDB transactions have evolved in three stages so far (Figure~\ref{fig:mongodb-sharded-cluster}):
In version 3.2, MongoDB used the WiredTiger storage engine as the default storage engine.
Utilizing the Multi-Version Concurrency Control (MVCC)
architecture of WiredTiger storage engine~\footnote{
  Snapshots and Checkpoints.
  \url{https://docs.mongodb.com/manual/core/wiredtiger/\#snapshots-and-checkpoints}
},
MongoDB was able to support single-document transactions
(called \wtalg) in the standalone deployment.
In version 4.0, MongoDB supported multi-document transactions (called \rsalg) in replica sets.
In version 4.2, MongoDB further introduced distributed (multi-document) transactions (called \scalg) in sharded clusters.
Each kind of MongoDB transactions has advantages and disadvantages.
In particular, distributed transactions
should not be a replacement for multi-document transactions
or single-document transactions,
since ``in most cases, (they) incur a greater performance cost
over single document writes.''~\footnote{
  Transactions and Atomicity.
  \url{https://docs.mongodb.com/manual/core/transactions/\#transactions-and-atomicity}}

A natural and fundamental question about MongoDB transactions is:
\emph{What transactional consistency guarantee do MongoDB transactions
in each deployment provide?}
This question poses three main challenges,
in terms of specification, protocols, and checking algorithms:
\begin{itemize}
  \item There are many variants of snapshot isolation
    in the literature~\cite{PODC2017:ClientCentric}.
    Though it was officially claimed that MongoDB implements a so-called
    speculative snapshot isolation protocol~\cite{VLDB2019:TunableConsistency},
    it is unclear which kind of MongoDB transactions in different deployments
    satisfies which specific variant of snapshot isolation.
  \item It lacks concise pseudocode of
    the transactional consistency protocols of MongoDB in different deployments,
    let alone the rigorous correctness proofs for them.
  \item Recently Biswas \emph{et al.} proved that the problem of
    checking whether a given history without the version order
    satisfies (Adya) snapshot isolation~\cite{Adya:PhDThesis1999}
    is \textsf{NP-complete}~\cite{OOPSLA2019:Complexity}.
    Therefore, it is challenging to efficiently check
    whether MongoDB in production satisfies
    some variant of snapshot isolation~\cite{VLDB2020:Elle}.
\end{itemize}

To answer the question above, we formally specify and verify
the transactional consistency protocols of MongoDB. Specifically,
\begin{itemize}
  \item We formally specify several variants of snapshot isolation
    in the well-known $(\vis, \ar)$ specification framework for transactional consistency models
    proposed by Cerone \emph{et al}.~\cite{CONCUR2015:Framework}.
  \item We provide a concise pseudocode
    for the transactional consistency protocols in each MongoDB deployment,
    based on the official documents and source code.
  \item We prove that \wtalg, \rsalg, and \scalg{} satisfy
    \strongsi, \rtsi, and \sessionsi, respectively (Figure~\ref{fig:mongodb-sharded-cluster}).
    In particular, \rtsi{} and \sessionsi{} are natural variants of snapshot isolation
    that we introduce for characterizing \rsalg{} and \scalg{}, respectively.
  \item We design white-box polynomial-time checking algorithms for
    the transactional protocols of MongoDB
    against \strongsi, \rtsi, and \sessionsi.
    These checking algorithms make use of the properties of the transactional protocols
    to infer the version order of histories,
    effectively circumventing the \textsf{NP-hard}
    obstacle in theory~\cite{OOPSLA2019:Complexity}.
    We then intensively test the transactional consistency protocols
    of MongoDB using Jepsen~\footnote{
      Distributed Systems Safety Research. \url{https://jepsen.io/}}.
    The results show that our checking algorithms are effective and efficiency.
\end{itemize}

The rest of the paper is organized as follows.
Section~\ref{section:si} formally specifies
several variants of snapshot isolation in the $(\vis, \ar)$
specification framework for transactional consistency models.
Sections~\ref{section:wt-tx}, \ref{section:rs-tx},
and \ref{section:sc-tx} describe the transactional consistency protocols
for MongoDB transactions in each deployment, respectively.
We also prove the correctness of them in the $(\vis, \ar)$ framework
in Section~\ref{section:app-proofs} of the Appendix.
Section~\ref{section:checking-si} proposes and evaluates
the efficient white-box checking algorithms
for MongoDB transactions in each deployment.
Section~\ref{section:related-work} discusses related work.
Section~\ref{section:conclusion} concludes the paper with possible future work.


\section{Snapshot Isolation} \label{section:si}

We consider a MongoDB deployment
$\deployment \in \deployments \triangleq \set{\wtdep, \rsdep, \scdep}$ managing a set of keys $\Key$, ranged over by $\keyvar$,
which take on a set of values $\Val$, ranged over by $\valvar$.
We denote by $\Op$ the set of possible read or write operations on keys:
$\Op = \set{\readevent(\keyvar, \valvar), \writeevent(\keyvar, \valvar)
  \mid \keyvar \in \Key, \valvar \in \Val}$.
We assume that each key $\keyvar$ has an dedicated initial value $\keyvar_{0} \in \Val$.
Each invocation of an operation is denoted by an event from a set $\Event$,
ranged over by $\evar$ and $\fvar$.
A function $\op : \Event \to \Op$ determines the operation a given event denotes.
Below we follow the $(\vis, \ar)$ specification framework proposed by~\cite{CONCUR2015:Framework, JACM2018:AnalysingSI}.
\subsection{Relations and Orderings} \label{ss:relations}

A binary relation $R$ over a given set $A$ is a subset of $A \times A$, i.e., $R \subseteq A \times A$.
For $a, b \in A$, we use $(a, b) \in R$ and $a \rel{R} b$ interchangeably.
The inverse relation of $R$ is denoted by $R^{-1}$,
i.e., $(a, b) \in R \iff (b, a) \in R^{-1}$.
We use $R^{-1}(b)$ to denote the set $\{ a \in A \mid (a, b) \in R \}$.
For some subset $A' \subseteq A$, the restriction of $R$ to $A'$ is
$R|_{A'} \triangleq R \cap (A' \times A)$.
For a relation $R$ and a deployment $\deployment \in \deployments$,
we use $R_{\deployment}$ to denote an instantiation of $R$ specific to this deployment $\deployment$.
Given two binary relations $R$ and $S$ over set $A$,
we define the composition of them as
$R \comp S = \{ (a, c) \mid \exists b \in A: a \rel{R} b \rel{S} c\}$.
A strict partial order is an irreflexive and transitive relation.
A total order is a relation which is a partial order and total.
\subsection{Histories and Abstract Executions} \label{ss:execution}

\begin{definition}[Transactions] \label{def:transaction}
  A \bfit{transaction} is a pair $(\E, \po)$,
  where $\E \subseteq \Event$ is a finite, non-empty set of events
  and $\po \subseteq \E \times \E$ is a strict total order called \bfit{program order}.
\end{definition}

For simplicity, we assume a dedicated transaction that writes initial values of all keys.
We also assume the existence of a time oracle that
assigns distinct real-time \emph{start} and \emph{commit} timestamps to each transaction $T$,
and access them by $\starttime(T)$ and $\committime(T)$, respectively.
For an transaction $T$ in deployment $\deployment \in \deployments$,
we use $\starttime_{\deployment}(T)$ (resp. $\committime_{\deployment}(T)$)
to denote the instantiation of $\starttime(T)$ (resp. $\committime(T)$) specific to the deployment $\deployment$.
We define two strict partial orders involving real time on transactions.

\begin{definition}[Returns Before] \label{def:rb}
  A transaction $S$ \bfit{returns before} $T$ in real time,
  denoted $S \rel{\rb} T$, if $\committime(S) < \starttime(T)$.
\end{definition}

\begin{definition}[Commits Before] \label{def:cb}
  A transaction $S$ \bfit{commits before} $T$ in real time,
  denoted $S \rel{\cb} T$, if $\committime(S) < \committime(T)$.
\end{definition}

Clients interact with MongoDB by issuing transactions via \emph{sessions}.
We use a \emph{history} to record the client-visible results of such interactions.
\begin{definition}[Histories] \label{def:history}
  A \bfit{history} is a pair $\h = (\T, \sorel)$,
  where $\T$ is a set of transactions with disjoint sets of events
  and the \bfit{session order} $\sorel \subseteq \T \times \T$
  is a union of strict total orders defined on disjoint sets of $\T$,
  which correspond to transactions in different sessions.
\end{definition}

To justify each transaction in a history,
we need to know how these transactions are related to each other.
This is captured declaratively
by the \emph{visibility} and \emph{arbitration} relations.

\begin{definition}[Abstract Executions] \label{def:ae}
  An \bfit{abstract execution} is a tuple $\ae = (\T, \sorel, \vis, \ar)$,
  where $(\T, \sorel)$ is a history,
  \bfit{visibility} $\vis \subseteq \T \times \T$ is a strict partial order,
  and \bfit{arbitration} $\ar \subseteq \T \times \T$
  is a strict total order such that $\vis \subseteq \ar$.
\end{definition}

For $\h = (\T, \sorel)$,
we often shorten $(\T, \sorel, \vis, \ar)$ to $(\h, \vis, \ar)$.


\begin{table}[tbp]
  \centering
  \caption{Consistency axioms, constraining an abstract execution $(\h, \vis, \ar)$.
    (Adapted from~\protect\cite{JACM2018:AnalysingSI})}
  \label{table:axioms}
  \renewcommand{\arraystretch}{1.8}
  \resizebox{\textwidth}{!}{%
    \begin{tabular}{|c|c|c|c|}
    \hline
    \multicolumn{4}{|c|}{
	  $\begin{aligned}[c]
		&\forall (\E, \po) \in \h.\;
		  \forall \evar \in \Event.\;
		    \forall \keyvar, \valvar.\;
			  \big(\op(\evar) = \readevent(\keyvar, \valvar)
			  \land \set{\fvar \mid (\op(\fvar) = \_(\keyvar, \_) \land \fvar \rel{\po} \evar} \neq \emptyset\big) \\
			  &\quad \implies \op(\max_{\po} \set{\fvar \mid \op(\fvar) = \_(\keyvar, \_) \land \fvar \rel{\po} \evar}) = \_(\keyvar, \valvar)
	   \end{aligned}$ \hfill (\intaxiom)} \\ \hline
    \multicolumn{4}{|c|}{
	  $\begin{aligned}[c]
		\forall T \in \h.\; \forall \keyvar, \valvar.\; T \vdash \readevent(\keyvar, \valvar) \implies
		  \max_{\ar}(\vis^{-1}(T) \cap \WriteTx_{\keyvar}) \vdash \writeevent(\keyvar, \valvar)
	  \end{aligned}$ \hfill (\extaxiom)} \\ \hline
    $\sorel \subseteq \vis$ \hfill (\sessionaxiom)
	& $\ar \comp \vis \subseteq \vis$ \hfill (\prefixaxiom)
	& \multicolumn{2}{c|}{$\forall S, T \in \h.\; S \conflict T \implies (S \rel{\vis} T \lor T \rel{\vis} S)$
	  \hfill (\noconflictaxiom)} \\ \hline
	$\rb \subseteq \vis$ \hfill (\rbaxiom)
	& $\vis \subseteq \rb$ \hfill (\inrbaxiom)
	& $\vis = \rb$ \hfill (\realtimesnapshotaxiom)
	& $\cb \subseteq \ar$ \hfill (\cbaxiom) \\ \hline
    \end{tabular}%
  }
\end{table}

\subsection{Consistency Axioms and (Adya) Snapshot Isolation} \label{ss:axioms}

A consistency model is a set $\axiom$ of \emph{consistency axioms} constraining abstract executions.
The model allows those histories for which
there exists an abstract execution that satisfies the axioms:
$\hist_{\axiom} = \set{\h \mid \exists \vis, \ar.\; (\h, \vis, \ar) \models \axiom}$.

We first briefly explain the consistency axioms
that are necessary for defining (Adya) snapshot isolation~\cite{Adya:PhDThesis1999},
namely \intaxiom, \extaxiom, \prefixaxiom,
and \noconflictaxiom~\cite{CONCUR2015:Framework,JACM2018:AnalysingSI}.
In Section~\ref{ss:si}, we will introduce a few new consistency axioms
and formally define several variants of snapshot isolation using them.
Table~\ref{table:axioms} summarizes all the consistency axioms used in this paper.
For $T = (\E, \po)$, we let $T \vdash \writeevent(\keyvar, \valvar)$
if $T$ writes to $\keyvar$ and the last value written is $\valvar$,
and $T \vdash \readevent(\keyvar, \valvar)$
if $T$ reads from $\keyvar$ before writing to it
and $\valvar$ is the value returned by the first such read.
We also use $\WriteTx_{\keyvar} = \set{T \mid T \vdash \writeevent(\keyvar, \_)}$.
Two transactions $S$ and $T$ conflicts, denoted $S \conflict T$,
if they write on the same key.

The \emph{internal consistency axiom} \intaxiom{} ensures that,
within a transaction, a read from a key returns the same value
as the last write to or read from this key in the transaction.
The \emph{external consistency axiom} \extaxiom{} ensures that
an external read in a transaction $T$ from a key
returns the value written by the last transaction in $\ar$
among all the transactions that proceed $T$ in terms of $\vis$ and write this key.
The \emph{prefix axiom} \prefixaxiom{} ensures that
if the snapshot taken by a transaction $T$ includes a transaction $S$,
than this snapshot also include all transactions that committed before $S$ in terms of $\ar$.
The \emph{no-conflict axiom} \noconflictaxiom{}
prevents concurrent transactions from writing on the same key.

The classic (Adya) snapshot isolation is then defined as follows.
\begin{definition}[$\si$~\cite{CONCUR2015:Framework}] \label{def:si}
  $\si = \intaxiom \land \extaxiom \land \prefixaxiom \land \noconflictaxiom$.
\end{definition}
\subsection{Variants of Snapshot Isolation} \label{ss:si}

There are several variants of snapshot isolation~\cite{PODC2017:ClientCentric},
including \ansisi~\cite{CritiqueIsolation:SIGMOD1995},
generalized snapshot isolation (\gsi)~\cite{SRDS2005:GSI},
strong snapshot isolation (\strongsi)~\cite{VLDB2006:LazyReplSI},
strong session snapshot isolation~\cite{VLDB2006:LazyReplSI},
parallel snapshot isolation (\parallelsi)~\cite{PSI:SOSP2011},
write snapshot isolation~\cite{EuroSys2012:CritiqueSI},
non-monotonic snapshot isolation (\nmsi)~\cite{NMSI:SRDS2013},
and prefix-consistent snapshot isolation~\cite{VLDB2006:LazyReplSI}.
We now formally define these SI variants in the $(\vis, \ar)$ framework
and illustrate their similarities and differences with examples.
For brevity, we concentrate on three variants that are concerned with MongoDB deployments,
namely \sessionsi, \rtsi, and \strongsi.
Particularly, the latter two variants are new.
More variants are explained in Appendix~\ref{section:app-si}.

Session snapshot isolation, denoted \sessionsi{},
requires a transaction to observe all the transactions that precedes it in its session.

\begin{definition}[$\sessionsi$~\cite{JACM2018:AnalysingSI}] \label{def:sessionsi}
  $\sessionsi = \si \land \sessionaxiom$.
\end{definition}


\begin{figure}[t]
  \centering
  \includegraphics[width = 0.5\textwidth]{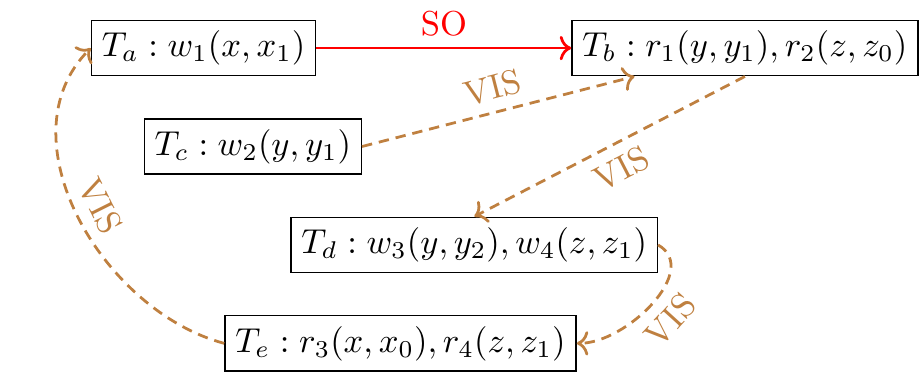}
  \caption{A history which satisfies \si{} but not \sessionsi{}; see Example~\ref{ex:si-not-sessionsi}.
    (Edges induced by transitivity are omitted.)}
  \label{fig:si-not-sessionsi}
\end{figure}

\begin{example}[\si{} vs. \sessionsi] \label{ex:si-not-sessionsi}
  Consider the history $\h$ in Figure~\ref{fig:si-not-sessionsi}.
  To show that $\h \models \si$,
  we construct an abstract execution $\ae = (\h, \vis, \ar)$
  such that $\vis$ is transitive, $T_{c} \rel{\vis} T_{b} \rel{\vis} T_{d} \rel{\vis} T_{e} \rel{\vis} T_{a}$ holds, and $\ar = \vis$.
  It is straightforward to justify that $\ae \models \si$.

  Now we show that $\h \not\models \sessionsi$.
  First, since $T_{e}$ read the initial value $x_{0}$ of key $x$
  and $T_{a}$ writes $x_{1}$ to $x$, we have $T_{e} \rel{\ar} T_{a}$.
  Similarly, $T_{b} \rel{\ar} T_{d}$.
  Second, since $T_{a} \rel{\sorel} T_{b}$, we have $T_{a} \rel{\vis} T_{b}$
  and thus $T_{a} \rel{\ar} T_{b}$.
  Third, since $T_{e}$ reads $z_{1}$ written by $T_{d}$,
  we have $T_{e} \rel{\vis} T_{d}$ and thus $T_{e} \rel{\ar} T_{d}$.
  In summary, $T_{a} \rel{\ar} T_{b} \rel{\ar} T_{d} \rel{\ar} T_{e} \rel{\ar} T_{a}$.
  By further exhaustive case analysis on the visibility between $T_{a}$ and $T_{e}$
  and between $T_{b}$ and $T_{d}$, we can show that no such $\vis$ and $\ar$ exist
  that $(\h, \vis, \ar) \models \sessionsi$.
\end{example}

Realtime snapshot isolation, denoted \rtsi{},
requires a transaction to observe all the transactions
that have returned before it starts in real time (i.e., \rbaxiom).
Moreover, it requires all transactions seem to be committed
in an order (specified by $\ar$) respecting
their commit order (specified by $\cb$) in real time (i.e., \cbaxiom).
Since $\sorel \subseteq \rb$, \rtsi{} is stronger than \sessionsi.

\begin{definition}[$\rtsi$] \label{def:rtsi}
  $\rtsi = \si \land \rbaxiom \land \cbaxiom$.
\end{definition}

In addition to \rtsi{}, strong snapshot isolation (denoted \strongsi{})
limits a transaction $T$ to read only from snapshots
that do \emph{not} include transactions that committed in real time \emph{after}
$T$ starts (i.e., \inrbaxiom).
Since $\realtimesnapshotaxiom = \rbaxiom \land \inrbaxiom$, we have

\begin{definition}[$\strongsi$] \label{def:strongsi}
  $\strongsi = \si \land \realtimesnapshotaxiom \land \cbaxiom$.
\end{definition}


\begin{figure}[t]
  \centering
  \includegraphics[width = 0.6\textwidth]{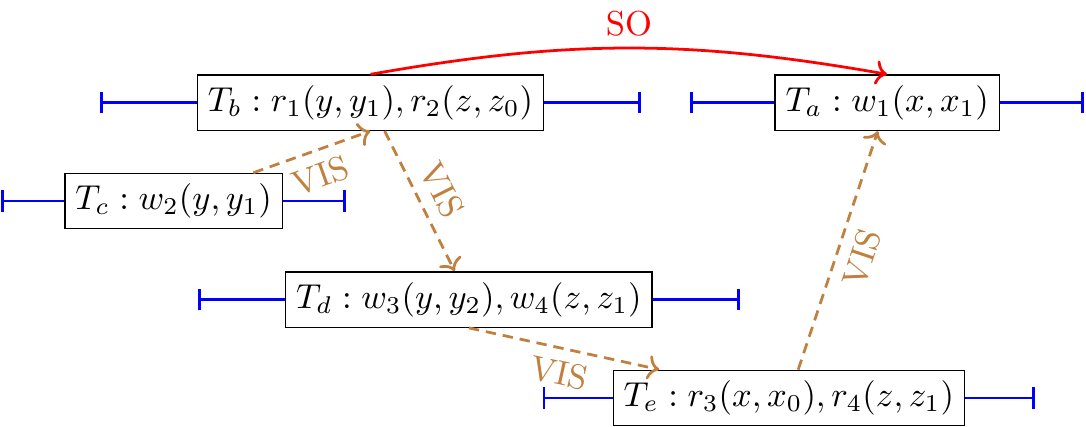}
  \caption{A history which satisfies \rtsi{} but not \strongsi{}; see Example~\ref{ex:rtsi-not-strongsi}.
    (The lifecycles of transactions are represented as intervals in blue.
     Time goes from left to right.)}
  \label{fig:rtsi-not-strongsi}
\end{figure}

\begin{example}[\rtsi{} vs. \strongsi] \label{ex:rtsi-not-strongsi}
  Consider the history $\h$ in Figure~\ref{fig:rtsi-not-strongsi}.
  To show that $\h \models \rtsi$, we construct an abstract execution $\ae = (\h, \vis, \ar)$
  such that $\vis$ is transitive, $T_{c} \rel{\vis} T_{b} \rel{\vis} T_{d} \rel{\vis} T_{e} \rel{\vis} T_{a}$ holds, and $\ar = \vis$.
  It is straightforward to justify that $\ae \models \rtsi$.
  On the other hand, $T_{b}$ reads $y_{1}$ from $y$ written by $T_{c}$,
  but $\lnot(T_{c} \rel{\rb} T_{b})$.
  Therefore, $\h \not\models \strongsessionsi$.
\end{example}


\section{WiredTiger Transactions} \label{section:wt-tx}


\begin{table}[t]
  \centering
  \caption{Types and variables used in \wtalg.}
  \label{table:wt-variables}
  \renewcommand{\arraystretch}{1.8}
  \resizebox{\textwidth}{!}{%
    \begin{tabular}{|c|c|c|c|c|}
      \hline
      \multicolumn{1}{|c|}{$\TID = \N \cup \set{-1, \wttidnone}$}
	    & \multicolumn{1}{c|}{$\currenttid \gets 1 \in \TID$}
	    & \multicolumn{1}{c|}{$\WTSID = \N$}
		& \multicolumn{1}{c|}{$\wtsession \in [\WTSID \to \WTTXN]$}
		& \multicolumn{1}{c|}{\rsstyle{$\maxcommitts \in \TS$}}
		\\ \hline
      \multicolumn{5}{|c|}{
		$\WTTXN = [\tid : \TID,
		  \upperlimit : \TID,
		  \concur : \ps(\TID), \mods : \ps(\Key \times \Val),
		  \rsstyle{\readts : \TS},
		  \rsstyle{\committs : \TS},
		  \scstyle{\preparets : \TS}]$
	  } \\ \hline
      \multicolumn{3}{|c|}{$\wttxnglobal \in [\WTSID \rightharpoonup [\tid : \TID, \rsstyle{\pinnedcommittedts : \TS}]]$}
        &\multicolumn{2}{c|}{$\store \in [\Key \to [\tid : \TID, \val : \Val,
		  \rsstyle{\ts : \TS},
		  \scstyle{\phase : \set{\wtprepareinprogress, \wtprepareresolved}}]^{\ast}]$} \\ \hline
    \end{tabular}%
  }
\end{table}

In this section, we describe the protocol \wtalg{} of snapshot isolation implemented in WiredTiger.
Table~\ref{table:wt-variables} summarizes the types and variables used in \wtalg.
For now the readers should ignore the \underline{underlined}
and \colorbox{lightgray}{highlighted} lines,
which are needed for \rsalg{} and \scalg{} protocols, respectively.
We reference pseudocode lines using the format algorithm\#:line\#.
For conciseness, we write something like $S \gets @ \oplus T$
to denote $S \gets S \oplus T$, where $\oplus$ is an operator.

\subsection{Key Designs} \label{ss:wt-designs}

\subsubsection{Transactions and the Key-Value Store} \label{sss:transactions}

We assume that each \wtalg{} transaction $\txnvar \in \WTTXN$ is associated with
a unique transaction identifier $\txnvar.\tid$
from set $\TID = \N \cup \set{-1, \wttidnone}$.
When a transaction $\txnvar$ starts, it initializes $\txnvar.\tid$ to $0$.
The actual (non-zero) $\txnvar.\tid$ is assigned
when its first \emph{update} operation is successfully executed.
$\currenttid$ tracks the next monotonically increasing transaction identifier to be allocated.
A transaction $\txnvar$ with $\txnvar.\tid \neq 0$ may be aborted
due to a conflict caused by a later update.
When this happens, we set $\txnvar.\tid = -1$.
Note that a read-only transaction $\txnvar$ always has $\txnvar.\tid = 0$.
We use dummy $\wttidnone$ to indicate that there is no such a transaction.

We model the key-value store, denoted $\store$,
as a function which maps each key $\keyvar \in \Key$
to a (possibly empty) list of pairs of the form $\tuple{\tidvar, \valvar}$
representing that the transaction with $\tidvar$
has written value $\valvar$ on $\keyvar$.
\subsubsection{Sessions} \label{sss:sessions}

Clients interact with \wtalg{} via sessions.
Each client is bind to a single session with a unique session identifier $\wtsidvar \in \WTSID$.
At most one transaction is active on a session at any time.
The mapping $\wtsession$ maintains the currently active transaction on each session
and $\wttxnglobal$ records which transaction has obtained its identifier on which session.
\subsubsection{The Visibility Rule} \label{sss:visibility}

To guarantee snapshot isolation, each transaction $\txnvar$ needs to
identify the set of transactions that are visible to it throughout its lifecycle when it starts.
Intuitively, each transaction $\txnvar$ is only aware of
all the transactions that have already been committed before it starts.
To this end, the transaction $\txnvar$ maintains
\begin{itemize}
  \item $\txnvar.\concur$: the set of identifiers of
        currently active transactions that have obtained their identifiers; and
  \item $\txnvar.\upperlimit$: the next transaction identifier (i.e., $\currenttid$) when $\txnvar$ starts.
\end{itemize}

The procedure \txnvis{} states that a transaction with $\tidvar$
is \emph{invisible} to another (active) transaction $\txnvar$ if
(line~\code{\ref{alg:wt-db}}{\ref{line:procedure-txnvis}})
\begin{itemize}
  \item it is aborted (i.e., $\tidvar = -1$), or
  \item it is concurrent with $\txn$ (i.e., $\tidvar \in \txnvar.\concur$), or
  \item it starts after $\txnvar$ and thus has a larger transaction identifier
        than $\txnvar.\upperlimit$
        (i.e., $\tidvar \ge \txnvar.\upperlimit$).
        Note that when $\txnvis$ is called,
        $\txnvar$ may have been assigned an identifier larger than $\txnvar.\upperlimit$.
        The second conjunction $\tidvar \neq \txnvar.\id$ allows $\txnvar$ to observe itself.
\end{itemize}

\begin{figure}[t]
  \centering
  \includegraphics[width = 0.5\textwidth]{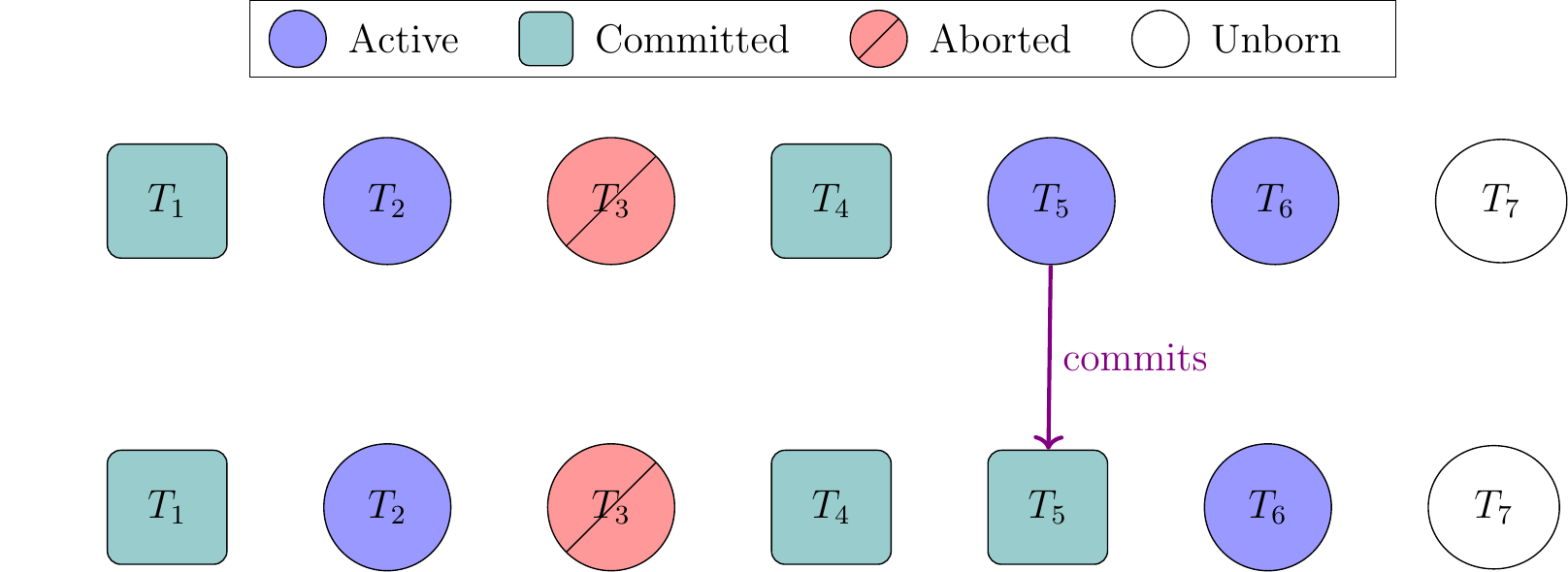}
  \caption{Illustration of the visibility rule; see Example~\ref{ex:visibility}.}
  \label{fig:ex:visibility}
\end{figure}

\begin{example}[Visibility] \label{ex:visibility}
  Consider the scenario in Figure~\ref{fig:ex:visibility},
  where transaction $T_{i}$ is on session $\wtsidvar_{i}$
  and would obtain its identifier $\tidvar_{i}$ if any.
  We assume that $\forall 1 \le i \le 6.\; \tidvar_{i} < \tidvar_{i+1}$.
  Suppose that when $T_{6}$ starts, (1) $T_{1}$ and $T_{4}$ have been committed;
  (2) $T_2$ and $T_5$ are active and have obtained their identifiers;
  (3) $T_3$ has been aborted; and
  (4) $T_{7}$ has not started.
  We have $T_{6}.\concur = \set{\tidvar_{2}, \tidvar_{5}}$
  and $T_{6}.\upperlimit = \tidvar_{6}$.
  According to the visibility rule, only $T_{1}$ and $T_{4}$ are visible to $T_{6}$.
  Note that even when $T_{5}$ commits later, it is still invisible to $T_{6}$.
\end{example}


\subsection{Protocol} \label{ss:wt-protocol}


\begin{algorithm}[p]
  \caption{\wtalg: the snapshot isolation protocol in WiredTiger}
  \label{alg:wt-db}
  \begin{varwidth}[t]{0.50\textwidth}
  \begin{algorithmic}[1]
    \Procedure{\wtstart}{$\wtsidvar$}
      \label{line:procedure-start}
      \State $\txnvar \gets \tuple{0, 0, \emptyset, \emptyset,
        \rsstyle{\wttsnone}, \rsstyle{\wttsnone}, \scstyle{\wttsnone}}$
        \label{line:wtstart-init-txn}
      \For{$(\_ \mapsto \tuple{\tidvar, \_}) \in \wttxnglobal \land \tidvar \neq \wttidnone$}
        \label{line:wtstart-forall-states}
        \State $\txnvar.\concur \gets @ \cup \set{\tidvar}$
          \label{line:wtstart-snapshot}
      \EndFor
      \State $\txnvar.\upperlimit \gets \currenttid$
        \label{line:wtstart-snapmax}
      \State $\wtsession[\wtsidvar] \gets \txnvar$
        \label{line:wtstart-txn}
    \EndProcedure

    \Statex
    \Procedure{\wtread}{$\wtsidvar, \keyvar$}
      \label{line:procedure-wtread}
      \State $\txnvar \gets \wtsession[\wtsidvar]$
        \label{line:wtread-txn}
      \For{$\tuple{\tidvar, \valvar, \rsstyle{\tsvar}, \scstyle{\phasevar}}
          \in \store[\keyvar]$} 
        \label{line:wtread-forall}
        \If{$\Call{\txnvis}{\txnvar, \tidvar, \rsstyle{\tsvar}}$}
          \label{line:wtread-call-txnvis}
            \State \Return $\tuple{\valvar, \scstyle{\phasevar}}$
              \label{line:wtread-return}
        \EndIf
      \EndFor
    \EndProcedure

    \Statex
    \Procedure{\wtupdate}{$\wtsidvar, \keyvar, \valvar$}
      \label{line:procedure-wtupdate}
      \State $\txnvar \gets \wtsession[\wtsidvar]$
        \label{line:wtupdate-txn}
      \For{$\tuple{\tidvar, \valvar, \rsstyle{\tsvar}, \scstyle{\_}}
          \in \store[\keyvar]$} 
        \label{line:wtupdate-forall}
        \If{$\lnot \Call{\txnvis}{\txnvar, \tidvar, \rsstyle{\tsvar}} \land \tidvar \neq -1$}
          \label{line:wtupdate-call-txnvis}
            \State $\Call{\wtrollback}{\wtsidvar}$
              \label{line:wtupdate-call-rollback}
            \State \Return \rollback
            \label{line:wtupdate-return-rollback}
        \EndIf
      \EndFor
      \If{$\txnvar.\tid = 0$}
        \label{line:wtupdate-txn-tid-if-wttidnone}
        \State $\txnvar.\tid \gets \currenttid$
          \label{line:wtupdate-txn-tid}
        \State $\currenttid \gets @ + 1$
          \label{line:wtupdate-wttxnglobal-currentid}
        \State $\wttxnglobal[\wtsidvar] \gets \tuple{\txnvar.\tid, \rsstyle{\wttsnone}}$
          \label{line:wtupdate-wttxnglobal-states}
      \EndIf
      \State $\txnvar.\mods \gets @ \cup \set{\tuple{\keyvar, \valvar}}$
        \label{line:wtupdate-txn-mods}
      \State $\wtsession[\wtsidvar] \gets \txnvar$
        \label{line:wtupdate-session-txn}
      \State $\store[\keyvar] \gets \tuple{\txnvar.\tid, \valvar, \rsstyle{\wttsnone}, \scstyle{\_}} \circ @$
        \label{line:wtupdate-store-key}
      \State \Return \ok
        \label{line:wtupdate-return}
    \EndProcedure

    \Statex
    \Procedure{\wtcommit}{$\wtsidvar$}
      \label{line:procedure-wtcommit}
      \State $\txnvar \gets \wtsession[\wtsidvar]$
        \label{line:wtcommit-txn}
      \For{$\tuple{\keyvar, \_} \in \txnvar.\mods$}
        \label{line:wtcommit-forall-mods}
        \For{$\tuple{\tidvar, \valvar, \rsstyle{\_}, {\scstyle{\_}}} \in \store[\keyvar]$}
          \label{line:wtcommit-forall-store}
          \If{$\tidvar = \txnvar.\tid$}
            \label{line:wtcommit-tid-if}
            \State $\rsstyle{\store[\keyvar] \gets \tuple{\tidvar, \valvar, \txnvar.\committs, {\scstyle{\_}}}}$
              \label{line:wtcommit-store}
          \EndIf
        \EndFor
      \EndFor
      \State $\wttxnglobal[\wtsidvar] \gets \tuple{\wttidnone, \rsstyle{\wttsnone}}$
        \label{line:wtcommit-wttxnglobal-states}
    \EndProcedure

    \Statex
    \Procedure{\wtrollback}{$\wtsidvar$}
      \label{line:procedure-rollback}
      \State $\txnvar \gets \wtsession[\wtsidvar]$
        \label{line:wtrollback-txn}
      \For{$\tuple{\keyvar, \_} \in \txnvar.\mods$}
        \label{line:wtrollback-forall-mods}
        \For{$\tuple{\tidvar, \valvar, \rsstyle{\tsvar}, {\scstyle{\phasevar}}} \in \store[\keyvar]$}
          \label{line:wtrollback-forall-store}
          \If{$\tidvar = \txnvar.\tid$}
            \label{line:wtrollback-tid-if}
            \State $\store[\keyvar] \gets \tuple{-1, \valvar, \rsstyle{\tsvar}, {\scstyle{\phasevar}}}$
              \label{line:wtrollback-store}
          \EndIf
        \EndFor
      \EndFor
      \State $\wttxnglobal[\wtsidvar] \gets \tuple{\wttidnone, \rsstyle{\wttsnone}}$
        \label{line:wtrollback-wttxnglobal-states}
    \EndProcedure
    \algstore{wt-db}
  \end{algorithmic}
  \end{varwidth}\qquad
  \begin{varwidth}[t]{0.50\textwidth}
  \begin{algorithmic}[1]
    \algrestore{wt-db}
    \Procedure{\txnvis}{$\txnvar, \tidvar, \rsstyle{\tsvar}$}
      \label{line:procedure-txnvis}
      \State \Return
          $\lnot \big(\tidvar = -1
          \lor \tidvar \in \txnvar.\concur
          \lor (\tidvar \ge \txnvar.\upperlimit \land \tidvar \neq \txnvar.\tid)\big)$
          $\rsstyle{\land\; \big(\tsvar \neq \wttsnone \land \tsvar \le \txnvar.\readts\big)}$
        \label{line:txnvis-invisible}
    \EndProcedure

    \Statex
    \Procedure{\rsstyle{\allcommitted}}{\null}
      \label{line:procedure-allcommitted}
      \State \Return the largest timestamp smaller than
        $\maxcommitts$ and $\min\set{\committsvar \mid \tuple{\_, \committsvar} \in \wttxnglobal \land \committsvar \neq \wttsnone}$
        \label{line:allcommitted-return}
    \EndProcedure

    \Statex
    \Procedure{\rsstyle{\wtsetreadts}}{$\wtsidvar, \readtsvar$}
      \State $\wtsession[\wtsidvar].\readts \gets \readtsvar$
        \label{line:wtsetreadts-readts}
    \EndProcedure

    \Statex
    \Procedure{\rsstyle{\wtsetcommitts}}{$\wtsidvar$}
      \label{line:procedure-wtsetcommitts}
      \State $\txnvar \gets \wtsession[\wtsidvar]$
        \label{line:wtsetcommitts-txn}
      \State $\committsvar \gets \txnvar.\committs$
        \label{line:wtsetcommitts-committs}
      \State $\maxcommitts \gets \max\set{@, \committsvar}$
        \label{line:wtsetcommitts-maxcommitts}
      \State $\wttxnglobal[\wtsidvar] \gets \tuple{\txnvar.\tid, \committsvar}$
        \label{line:wtsetcommitts-wttxnglobal-states}
      \State \Return \ok
        \label{line:wtsetcommitts-return}
    \EndProcedure


    \Statex
    \Procedure{\scstyle{\wtprepare}}{$\wtsidvar, \preparetsvar$}
      \label{line:procedure-wtprepare}
      \State $\txnvar \gets \wtsession[\wtsidvar]$
        \label{line:wtprepare-txn}
      \State $\txnvar.\preparets \gets \preparetsvar$
        \label{line:wtprepare-txn-preparets}
      \For{$\tuple{\keyvar, \_} \in \txnvar.\mods$}
        \label{line:wtprepare-forall-mods}
        \For{$\tuple{\tidvar, \valvar, \_, \_} \in \store[\keyvar]$}
          \label{line:wtprepare-forall-store}
          \If{$\tidvar = \txnvar.\tid$}
            \label{line:wtprepare-tid-if}
            \State $\store[\keyvar] \gets \langle\tidvar, \valvar,
              \txnvar.\preparets, \wtprepareinprogress\rangle$
              \label{line:wtprepare-store}
          \EndIf
        \EndFor
      \EndFor
      \State $\wttxnglobal[\wtsidvar] \gets \tuple{\wttidnone, \wttsnone}$
        \label{line:wtprepare-wttxnglobal-states}
      \State \Return \ok
        \label{line:wtprepare-return}
    \EndProcedure

    \Statex
    \Procedure{\scstyle{\wtcommitpreparets}}{$\wtsidvar, \committsvar$}
      \State $\txnvar \gets \wtsession[\wtsidvar]$
        \label{line:wtcommitpreparets}
      \State $\txnvar.\committs \gets \committsvar$
        \label{line:wtcommitpreparets-committs}
      \State $\wttxnglobal[\wtsidvar] \gets \tuple{\txnvar.\tid, \txnvar.\committs}$
        \label{line:wtcommitpreparets-wttxnglobal-states}
    \EndProcedure

    \Statex
    \Procedure{\scstyle{\wtcommitprepare}}{$\wtsidvar$}
      \label{line:procedure-wtcommitprepare}
      \State $\txnvar \gets \wtsession[\wtsidvar]$
        \label{line:wtcommitprepare-txn}
      \For{$\tuple{\keyvar, \_} \in \txnvar.\mods$}
        \label{line:wtcommitprepare-forall-mods}
        \For{$\tuple{\tidvar, \valvar, \_, \_} \in \store[\keyvar]$}
          \label{line:wtcommitprepare-forall-store}
          \If{$\tidvar = \txnvar.\tid$}
            \label{line:wtcommitprepare-tid-if}
            \State $\store[\keyvar] \gets \langle\tidvar, \valvar,
              \txnvar.\committs, \wtprepareresolved \rangle$
              \label{line:wtcommitprepare-store}
          \EndIf
        \EndFor
      \EndFor
      \State $\wttxnglobal[\wtsidvar] \gets \tuple{\wttidnone, \wttsnone}$
        \label{line:wtcommitprepare-wttxnglobal-states-2}
      \State $\maxcommitts \gets \max\set{@, \txnvar.\committs}$
        \label{line:wtcommitprepare-wttxnglobal-maxcommitts}
    \EndProcedure
  \end{algorithmic}
  \end{varwidth}
\end{algorithm}

For simplicity, we assume that each handler in the protocols executes atomically;
see Section~\ref{section:conclusion} for discussions.
\subsubsection{Start Transactions} \label{sss:wt-open-start}

A client starts a transaction on a session $\wtsidvar$ by calling \wtstart{},
which creates and populates a transaction $\txnvar$
(lines~\code{\ref{alg:wt-db}}{\ref{line:wtstart-init-txn}}--\code{\ref{alg:wt-db}}{\ref{line:wtstart-snapmax}}).
Particularly, it scans $\wttxnglobal$ to collect the concurrently active transactions on other sessions into $\txnvar.\concur$.
\subsubsection{Read and Update Operations} \label{sss:wt-read-update}

To read from a key $\keyvar$,
we iterate over the update list $\store[\keyvar]$ forward
and returns the value written by the \emph{first} visible transaction
(line~\code{\ref{alg:wt-db}}{\ref{line:wtread-call-txnvis}}).

To update a key $\keyvar$,
we first check whether the transaction, denoted $\txnvar$,
should be aborted due to conflicts
(lines~\code{\ref{alg:wt-db}}{\ref{line:wtupdate-forall}}--\code{\ref{alg:wt-db}}
{\ref{line:wtupdate-return-rollback}}).
To this end, we iterates over the update list $\store[\keyvar]$.
If there are updates on $\keyvar$ made by transactions
that are \emph{invisible} to $\txnvar$ and are not aborted,
$\txnvar$ will be rolled back.
If $\txnvar$ passes the conflict checking,
it is assigned a unique transaction identifier, i.e., $\currenttid$,
in case it has not yet been assigned one
(line~\code{\ref{alg:wt-db}}{\ref{line:wtupdate-txn-tid}}).
Finally, the key-value pair $\tuple{key, val}$
is added into the modification set $\txnvar.\mods$
and is inserted at the front of the update list $\store[\keyvar]$.
\subsubsection{Commit or Rollback Operations} \label{sss:wt-commit-rollback}

To commit the transaction on session $\wtsidvar$,
we simply resets $\wttxnglobal[\wtsidvar]$ to $\wttidnone$,
indicating that there is currently no active transaction
on this session~(line~\code{\ref{alg:wt-db}}{\ref{line:wtcommit-wttxnglobal-states}}).
To roll back a transaction $\txnvar$,
we additionally reset $\txnvar.\tid$ in $\store$ to $-1$
(line~\code{\ref{alg:wt-db}}{\ref{line:wtrollback-store}}).
Note that read-only transactions can always commit successfully.



\section{Replica Set Transactions} \label{section:rs-tx}


\begin{table}[t]
  \centering
  \caption{Types and variables used in \rsalg.}
  \label{table:rs-variables}
  \renewcommand{\arraystretch}{1.5}
  \resizebox{1.00\textwidth}{!}{%
    \begin{tabular}{|c|c|c|c|c|}
      \hline
      \multicolumn{1}{|c|}{$\RSSID = \N$}
	    & \multicolumn{1}{c|}{$\TS$: the set of timestamps}
	    & \multicolumn{1}{c|}{$\ct \in \TS$}
      & \multicolumn{1}{c|}{$\oplog \gets \emptyseq \in \OPLOG^{\ast}$}
	    & \multicolumn{1}{c|}{$\rswtconns \in [\RSSID \to \WTSID]$}
      \\ \hline
      \multicolumn{3}{|c|}{$\OPLOG = [\ts : \TS, \ops : (\Key \times \Val)^{\ast}]
        \cup \scstyle{[\ts : \TS, {\sf cts} : \TS]}$}
	    & \multicolumn{2}{c|}{$\txnops \in [\RSSID \to (\Key \times \Val)^{\ast}]$}
      \\ \hline
    \end{tabular}%
  }
\end{table}

We now describe the protocol \rsalg{} of snapshot isolation implemented in replica sets.
Table~\ref{table:rs-variables} summarizes the types and variables used in \rsalg.
For now the readers should ignore the \colorbox{lightgray}{highlighted} lines,
which are needed for the \scalg{} protocol.


\subsection{Key Designs} \label{ss:rs-designs}

A replica set consists of a single primary node and several secondary nodes.
All transactional operations, i.e., start, read, update, and commit, are first performed on the primary.
Committed transactions are replicated to the secondaries via a leader-based consensus protocol
similar to Raft~\cite{ATC2014:Raft, NSDI2021:MongoDB-Raft}.
\subsubsection{Hybrid Logical Clocks} \label{sss:rs-hlc}

\rsalg{} uses hybrid logical clocks (HLC)~\cite{HLC:2014}
as the read and commit timestamps of transactions.
Without going into the details,
we assume that HLCs are compared lexicographically and are thus totally ordered,
and HLCs can be incremented via \tick{}.

All nodes and clients maintain a cluster time $\ct$, which is also an HLC~\cite{SIGMOD2019:MongoDB-CC}.
They \emph{distribute} their latest cluster time when sending any messages
and \emph{update} it when receiving a larger one in incoming messages~\footnote{
  For brevity, we omit the distribution of cluster time in pseudocode.}.
The cluster time is \emph{incremented} (``ticks'') only when
a \rsalg{} transaction is committed on the primary.
\subsubsection{Speculative Snapshot Isolation} \label{sss:ssi}

\rsalg{} implements a so-called \emph{speculative snapshot isolation} protocol
with \rc = ``\snapshotrc'' and \wc = ``\majority''~\cite{VLDB2019:TunableConsistency}.
It guarantees that each read obtains data that was majority committed in the replica set,
and requires all updates be majority committed before the transaction completes.
To reduce aborts due to conflicts in back-to-back transactions,
\rsalg{} adopts an innovative strategy called ``\emph{speculative majority}''~\cite{VLDB2019:TunableConsistency}.
In this strategy, transactions read the \emph{latest} data,
instead of reading at a timestamp at or earlier than
the majority commit point in WiredTiger.
At commit time, they wait for all the data they read to
become majority committed.

In implementation, it is unnecessary for update transactions
to explicitly wait at commit time for the data read
to become majority committed~\cite{VLDB2019:TunableConsistency}.
This is because we must wait for the updates in those transactions
to be majority committed, which, due to the replication mechanism (Section~\ref{sss:replication}),
implies that the data read was also majority committed.
Read-only transactions, however, need to issue a special ``\noop'' operation
at commit time and wait for it to be majority committed.
\subsubsection{Read and Commit Timestamps} \label{sss:rs-oplog}

The primary node maintains an $\oplog$ of transactions,
where each entry is assigned a unique \emph{commit timestamp}.
These commit timestamps determine the (logical) commit order of \rsalg{} transactions,
no matter when they are encapsulated into \wtalg{} transactions and are committed in WiredTiger.

When a transaction starts, it is assigned a \emph{read timestamp} on the primary
such that all transactions with smaller commit timestamps have been committed in WiredTiger.
That is, the read timestamp is the maximum point at which the oplog of the primary has no gaps.
Specifically, in $\allcommitted$
(line~\code{\ref{alg:wt-db}}{\ref{line:procedure-allcommitted}}),
the read timestamp of a transaction $\txnvar$ is computed as
the largest timestamp smaller than the minimum of $\maxcommitts$
and the set of commit timestamps of transactions concurrent with $\txnvar$,
where $\maxcommitts$ is the maximum commit timestamp that WiredTiger knows.
\subsection{Protocol} \label{ss:rs-protocol}


\begin{algorithm}[t]
  \caption{\rsalg: the snapshot isolation protocol in a replica set (the primary node)}
  \label{alg:rs-primary}
  \begin{varwidth}[t]{0.50\textwidth}
    \begin{algorithmic}[1]
      \Procedure{\rsstart}{$\rssidvar$}
        \label{line:procedure-rsstart}
      \State Set \emph{readConcern} = ``\emph{snapshot}'' and \emph{writeConcern} = ``majority''
        \label{line:rsstart-concern}
      \EndProcedure

      \Statex
      \Procedure{\rsread}{$\rssidvar, \keyvar$}
        \label{line:procedure-rsread}
      \State $\openwtsession(\rssidvar)$
        \label{line:rsread-call-openwtsession}
      \State $\wtsidvar \gets \rswtconns[\rssidvar]$
        \label{line:rsread-wtsession}
      \Repeat
        \State $\tuple{\valvar, \scstyle{\phasevar}} \gets \Call{\wtread}{\wtsidvar, \keyvar}$
          \label{line:rsread-call-wtread}
      \Until{$\scstyle{\phasevar \neq \wtprepareinprogress}$}
        \label{line:rsread-while-status}
        \label{line:rsread-call-wtread-in-while}
      \State \Return $\valvar$
        \label{line:rsread-return}
      \EndProcedure

      \Statex
      \Procedure{\rsupdate}{$\rssidvar, \keyvar, \valvar$}
      \label{line:procedure-rsupdate}
      \State $\openwtsession(\rssidvar)$
      \label{line:rsupdate-call-openwtsession}
      \State $\scstyle{\Call{\rsread}{\rssidvar, \keyvar}}$
      \label{line:rsupdate-call-rsread}
      \State $\wtsidvar \gets \rswtconns[\rssidvar]$
      \label{line:rsupdate-wtsidvar}
      \State $\statusvar \gets$ $\Call{\wtupdate}{\wtsidvar, \keyvar, \valvar}$
      \label{line:rsupdate-call-wtupdate}
      \If{$\statusvar = \ok$}
      \label{line:rsupdate-success}
      \State $\txnops[\rssidvar] \gets @ \circ \set{\tuple{\keyvar, \valvar}}$
      \label{line:rsupdate-txnops}
      \EndIf
      \State \Return $\statusvar$
      \label{line:rsupdate-return}
      \EndProcedure

      \Statex
      \Procedure{\rsrollback}{$\rssidvar$}
        \label{line:procedure-rsrollback}
        \State $\Call{\wtrollback}{\rswtconns[\rssidvar]}$
          \label{line:rsrollback-call-wtrollback}
        \EndProcedure
      \algstore{rs-primary}
    \end{algorithmic}
  \end{varwidth}\qquad
  \begin{varwidth}[t]{0.50\textwidth}
    \begin{algorithmic}[1]
      \algrestore{rs-primary}
      \Procedure{\rscommit}{$\rssidvar$}
        \label{line:procedure-rscommit}
        \State \atomicbegin $\ct \gets \Call{tick}{\null}$
          \label{line:rscommit-tick}
        \State $\wtsidvar \gets \rswtconns[\rssidvar]$
          \label{line:rscommit-wtsession}
        \State $\wtsession[\wtsidvar].\committs \gets \ct$
          \label{line:rscommit-txn-committs}
        \State $\Call{\wtsetcommitts}{\wtsidvar}$\atomicend
          \label{line:rscommit-call-wtsetcommitts}

        \State $\opsvar \gets \txnops[\rssidvar]$
          \label{line:rscommit-ops}
        \If{$\opsvar = \emptyset$}
          \label{line:rscommit-ops-empty}
          \State $\oplog \gets \oplog \circ \tuple{\ct, \noop}$
            \label{line:rscommit-oplog-noop}
        \Else
        \State $\oplog \gets \oplog \circ \tuple{\ct, \opsvar}$
          \label{line:rscommit-oplog}
        \EndIf
        \State $\Call{\wtcommit}{\wtsidvar}$
          \Comment{\cyan{locally committed}}
          \label{line:rscommit-call-wtcommit}

        \State \wait $\lastmajoritycommitted \ge \ct$
          \Comment{\cyan{majority committed or simply committed}}
          \label{line:rscommit-wait-majoritycommitted}
        \State \Return \ok
          \label{line:rscommit-return}
      \EndProcedure

      \Statex
      \Procedure{\openwtsession}{$\rssidvar$}
      \label{line:procedure-openwtsession}
      \If{it is the first operation of the transaction}
      \label{line:openwtsession-if-starttxn}
      \State $\rswtconns[\rssidvar] \gets \text{a new } \wtsidvar$
      \label{line:openwtsession-rswtconns}
      \State $\Call{\wtstart}{\wtsidvar}$
      \label{line:openwtsession-call-wtstart}
      \State $\readtsvar \gets \Call{\allcommitted}{\null}$
      \label{line:openwtsession-call-allcommitted}
      \State $\Call{\wtsetreadts}{\wtsidvar, \readtsvar}$
      \label{line:openwtsession-call-wtsetreadts}
      \EndIf
      \EndProcedure
    \end{algorithmic}
  \end{varwidth}
\end{algorithm}


Clients interacts with \rsalg{} via sessions.
Each client is bind to a single session with a unique session identifier $\rssidvar \in \RSSID$,
and at most one transaction is active on a session at any time.
Each active \rsalg{} transaction on a session $\rssidvar$
is encapsulated into a \wtalg{} transaction on a new session $\wtsidvar$,
as recorded in $\rswtconns$.
\subsubsection{Read and Update Operations} \label{sss:rs-read-update}

When the primary receives the \emph{first} operation of an transaction
(lines~\code{\ref{alg:rs-primary}}{\ref{line:rsread-call-openwtsession}}
and \code{\ref{alg:rs-primary}}{\ref{line:rsupdate-call-openwtsession}}),
it calls $\openwtsession$ to open a new session $\wtsidvar$ to WiredTiger,
start a new \wtalg{} transaction on $\wtsidvar$,
and more importantly set the transaction's read timestamp.

The primary delegates the read/update operations to WiredTiger
(lines~\code{\ref{alg:rs-primary}}{\ref{line:rsread-call-wtread}}
and \code{\ref{alg:rs-primary}}{\ref{line:rsupdate-call-wtupdate}}).
If an update succeeds, the $\tuple{key, val}$ pair is recorded in $\txnops[\rssidvar]$
(line~\code{\ref{alg:rs-primary}}{\ref{line:rsupdate-txnops}}).
\subsubsection{Commit Operations} \label{sss:rs-commit}

To commit a transaction, the primary first \emph{atomically} increments its cluster time $\ct$ via \tick,
takes it as the transaction's commit timestamp
(line~\code{\ref{alg:rs-primary}}{\ref{line:rscommit-txn-committs}}),
uses it to update $\maxcommitts$, and records it in $\wttxnglobal$
(lines~\code{\ref{alg:rs-primary}}{\ref{line:rscommit-call-wtsetcommitts}}
and~\code{\ref{alg:wt-db}}{\ref{line:procedure-wtsetcommitts}}).

If this is a read-only transaction,
the primary appends a $\noop$ entry to its $\oplog$
(line~\code{\ref{alg:rs-primary}}{\ref{line:rscommit-oplog-noop}}; Section~\ref{sss:ssi}).
Otherwise, it appends an entry containing the updates of the transaction.
Each oplog entry is associated with the commit timestamp of the transaction.
Then, the primary asks WiredTiger to \emph{locally commit} this transaction in \wtcommit{}
(line~\code{\ref{alg:rs-primary}}{\ref{line:rscommit-call-wtcommit}}),
which associates the updated key-value pairs in $\store$ with the commit timestamp
(line~\code{\ref{alg:wt-db}}{\ref{line:wtcommit-store}}).
Note that \wtcommit{} needs not to be atomically executed with $\tick$ and $\wtsetcommitts$.

Finally, the primary waits for all updates of the transaction to be \emph{majority committed}
(line~\code{\ref{alg:rs-primary}}{\ref{line:rscommit-wait-majoritycommitted}}).
Specifically, it waits for $\lastmajoritycommitted \ge \ct$,
where $\lastmajoritycommitted$ is the timestamp of the last oplog entry
that has been majority committed (discussed shortly).

\begin{table}[t]
  \centering
  \caption{Types and variables used in \scalg.}
  \label{table:sc-variables}
  \renewcommand{\arraystretch}{1.3}
  \resizebox{0.60\textwidth}{!}{%
    \begin{tabular}{|c|c|c|}
      \hline
      \multicolumn{1}{|c|}{$\SCSID = \N$}
	    & \multicolumn{1}{c|}{$\ShardID = \N$}
	    & \multicolumn{1}{c|}{$\scrsconns \in [\SCSID \to \RSSID]$}
      \\ \hline
      \multicolumn{2}{|c|}{$\shards \in [\SCSID \to \ps(\ShardID)]$}
      & \multicolumn{1}{c|}{$\readts \in [\SCSID \to \TS]$}
      \\ \hline
    \end{tabular}%
  }
\end{table}

\subsubsection{Replication} \label{sss:replication}


\begin{algorithm}[t]
  \caption{\rsalg: replication in a replica set}
  \label{alg:replication}
  \begin{algorithmic}[1]
    \Procedure{\replicate}{\null}
      \label{line:procedure-replicate}
      \State $\send \Call{\pulloplog}{\lastpulled} \sendto \text{the primary } \primaryvar$
        \label{line:replicate-call-pulloplog}

      \State $\wait\receive \Call{\pushoplog}{\oplogvar, \ctvar} \from \primaryvar$
        \label{line:replicate-receive-pushoplog}
      \State $\oplog \gets @ \circ \oplogvar$
        \label{line:replicate-oplog}
      \State $\lastpulled \gets \ctvar$
        \label{line:replicate-lastpulled}
      \State $\send \Call{\replicateack}{\lastpulled} \sendto \primaryvar$
        \label{line:replicate-send-replicateack}
    \EndProcedure

    \WhenReceive[$\pulloplog(\lastpulledvar)$]\from $\secondaryvar$
      \label{line:procedure-pulloplog}
      \State $\oplogvar \gets \text{ oplog entries after $\lastpulledvar$ in } \oplog$
        \label{line:pulloplog-oplog}
      \State $\send \Call{\pushoplog}{\oplogvar, \ct} \sendto \secondaryvar$
        \label{line:pulloplog-send-pushoplog}
    \EndWhenReceive

    \WhenReceive[$\replicateack(\lastpulledvar)$]\from $\secondaryvar$
      \label{line:procedure-replicateack}
      \State $\lastpulledacks[\secondaryvar] \gets \lastpulledvar$
        \label{line:replicateack-lastpulledacks}
      \State $\lastmajoritycommitted \gets \text{the } \lfloor n/2 \rfloor\text{-th largest timestamp in } \lastpulledacks$
        \label{line:replicateack-lastmajoritycommitted}
    \EndWhenReceive
  \end{algorithmic}
\end{algorithm}


Each secondary node periodically pulls the oplog entries
with larger commit timestamps than $\lastpulled$ from the primary
(line~\code{\ref{alg:replication}}{\ref{line:replicate-call-pulloplog}}).
It appends the retrieved entries to its own $\oplog$,
updates $\lastpulled$ accordingly,
and sends a \replicateack{} to the primary.
The primary maintains the timestamp of the last pulled oplog entry
for each secondary $\secondaryvar$ in $\lastpulledacks[\secondaryvar]$.
Whenever it receives a \replicateack{} from a secondary,
the primary updates $\lastmajoritycommitted$ accordingly
(line~\code{\ref{alg:replication}}{\ref{line:replicateack-lastmajoritycommitted}}).




\section{Sharded Cluster Transactions} \label{section:sc-tx}

This section describes the protocol \scalg{} of snapshot isolation implemented in sharded clusters.
Table~\ref{table:sc-variables} summarizes the types and variables used in \scalg.


\subsection{Key Designs} \label{ss:sc-key-designs}


A client issues distributed transactions via a session connected to a \emph{mongos}.
The mongos, as a transaction router,
uses its cluster time as the \emph{read timestamp} of the transaction
and forwards the transactional operations to corresponding shards.
The shard which receives the first read/update operation of a transaction
is designated as the transaction coordinator.
\subsubsection{Two Phase Commit among Shards} \label{sss:sc-2pc}

If a transaction has not been aborted due to write conflicts in \scupdate{},
the mongos can proceed to commit it.
If this transaction is read-only,
the mongos instructs each of the participants to directly commit locally via \rscommit;
otherwise, the mongos instructs the transaction coordinator
to perform a variant of two-phase commit (2PC) that always commits among all participants
(line~\code{\ref{alg:sc-primary}}{\ref{line:procedure-twopc}}).
Specifically, the coordinator sends a $\prepare$ message to all participants.
After receiving the $\prepare$ message,
a participant computes a local \emph{prepare timestamp}
and returns it to the coordinator in a $\prepareack$ message.
When the coordinator receives $\prepareack$ messages from all participants,
it calculates the transaction's \emph{commit timestamp}
by taking the maximum of all prepare timestamps
(line~\code{\ref{alg:sc-primary}}{\ref{line:twopc-committs}}),
and sends a $\commit$ message to all participants.
After receiving $\decack$ messages from all participants,
the coordinator replies to the mongos
(line~\code{\ref{alg:sc-primary}}{\ref{line:twopc-call-twopcack}}).
\subsubsection{Replication within Replia Sets} \label{sss:sc-replication}

\scalg{} uses state-machine replication~\cite{StateMachine:CSUR1990,TCS:DC2021}
to achieve fault tolerance in replica sets.
On the one hand, the transaction coordinator
persists the participant information within its replica set
(line~\code{\ref{alg:sc-primary}}{\ref{line:twopc-participants-majority-committed}})
before sending the $\prepare$ messages,
and the transaction's commit timestamp
(line~\code{\ref{alg:sc-primary}}{\ref{line:twopc-decision-majority-committed}})
before sending the \commit{} messages.
On the other hand, the primary of a shard waits for a quorum of secondary nodes
to persist its oplog entry
(line~\code{\ref{alg:sc-primary}}{\ref{line:prepare-wait-majoritycommitted}})
before sending the $\prepareack$ message,
and the final decision
(lines~\code{\ref{alg:sc-primary}}{\ref{line:commit-wait}} and \code{\ref{alg:sc-primary}}{\ref{line:abort-wait}})
before sending the \decack{} message.
\subsubsection{Consistent Snapshots and Read Timestamps} \label{sss:sc-snapshot}

\scalg{} uses HLCs which are loosely synchronized to assign read and commit timestamps to transactions.
Due to clock skew or pending commit,
a transaction may receive a read timestamp from a mongos,
but the corresponding snapshot is not yet fully available
at transaction participants~\cite{SRDS2013:ClockSI}.
Therefore, \scalg{} will delay the read/update operations until the snapshot becomes available.
There are four cases~\cite{SRDS2013:ClockSI}.
\begin{enumerate}[(1)]
	\item (\caseclockskew)
	      When a transaction participant
	      receives a read/update operation forwarded by the mongos
	      and finds that its cluster time is behind the read timestamp,
	      it first increments its cluster time to catch up.
	      \scalg{} achieves this by issuing a $\noop$ write with the read timestamp
	      (line~\code{\ref{alg:sc-primary}}{\ref{line:waitforreadconcern-call-noopwrite}}).
	\item (\caseholes)
	      \scalg{} transactions on a primary are not necessarily committed
	      in the increasing order of their commit timestamps in WiredTiger.
	      To guarantee snapshot isolation,
	      we need to ensure that there are no ``holes'' in the oplog before the read timestamp.
	      Therefore, the primary waits until $\Call{allcommitted}{\null}$
		  is larger than the read timestamp
	      (line~\code{\ref{alg:sc-primary}}{\ref{line:waitforreadconcern-wait-until-allcommitted}}).
	\item (\casependingcommitread)
	      Consider a read operation of a transaction $\txnvar$.
	      By the visibility rule in WiredTiger,
	      the read may observe an update of another transaction $\txnvar'$ which is prepared but not yet committed.
	      To guarantee snapshot isolation,
	      the read cannot be applied until $\txnvar'$ commits or aborts.
	      To this end, \scalg{} keeps trying the read until it returns a value without the $\wtprepareinprogress$ flag
	      (line~\code{\ref{alg:rs-primary}}{\ref{line:rsread-call-wtread-in-while}}).
	\item (\casependingcommitupdate)
	      Similarly, an update operation of a transaction
	      may observe an update of another transaction which is prepared but not yet committed.
	      \scalg{} delays this update by first performing a read operation on the same key
	      in the way described in \casependingcommitread{}
	      (line~\code{\ref{alg:rs-primary}}{\ref{line:rsupdate-call-rsread}}).
\end{enumerate}
\subsection{Protocols} \label{ss:sc-protocols}


\begin{algorithm}[t]
  \caption{\scalg: the snapshot isolation protocol in sharded cluster (on the primary nodes)}
  \label{alg:sc-primary}
  \begin{varwidth}[t]{0.50\textwidth}
  \begin{algorithmic}[1]
	\WhenReceive[$\Call{\scread}{\scsidvar, \keyvar}$]
	  \label{line:procedure-scread}
    \State $\Call{\scstart}{\scsidvar}$
      \label{line:scread-call-scstart}
    \State $\valvar \gets \Call{\rsread}{\scrsconns[\scsidvar], \keyvar}$
      \label{line:scread-call-rsread}
    \State \Return $\valvar$
      \label{line:scread-return}
	\EndWhenReceive

  \Statex
  \WhenReceive[$\Call{\scupdate}{\scsidvar, \keyvar, \valvar}$]
	  \label{line:procedure-scupdate}
    \State $\Call{\scstart}{\scsidvar}$
      \label{line:scupdate-call-scstart}
    \State $\statusvar \gets \Call{\rsupdate}{\scrsconns[\scsidvar], \keyvar, \valvar}$
      \label{line:scupdate-call-rsupdate}
    \State \Return $\statusvar$
      \label{line:scupdate-return}
  \EndWhenReceive

  \Statex
  \WhenReceive[$\Call{\twopc}{\scsidvar} \from \mongosvar$]
    \label{line:procedure-twopc}
    \State \untilstyle{\wait\until $\tuple{\scsidvar, \shards[\scsidvar]}$ majority committed
      in collection \texttt{config.transaction\_coords}}
      \label{line:twopc-participants-majority-committed}

    \State $\Primary \gets \primaryof(\shards[\scsidvar])$
      \label{line:twopc-primary}
    \State $\send \Call{\prepare}{\scsidvar} \sendto \Primary$
      \label{line:twopc-call-prepare}
    \State $\wait\receive \Call{\prepareack}{\preparetsvar_{\primaryvar}}$
      \Statex \hspace{25pt} $\from \primaryvar \in \Primary$ 
      \label{line:twopc-wait-receive-prepareack}
    \State $\committsvar \gets \max_{\primaryvar \in \Primary} \preparetsvar_{\primaryvar}$
      \label{line:twopc-committs}

    \State \untilstyle{\wait\until $\tuple{\scsidvar, \committsvar}$ majority committed
      in collection \texttt{config.transaction\_coords}} 
      \label{line:twopc-decision-majority-committed}

    \State $\send \Call{\commit}{\scsidvar, \committsvar} \sendto \Primary$
      \label{line:twopc-call-commit}
    \State $\wait\receive \Call{\decack}{\scsidvar} \from \Primary$
      \label{line:twopc-wait-receive-decack}
    \State $\send \Call{\twopcack}{\null} \sendto \mongosvar$
      \label{line:twopc-call-twopcack}
  \EndWhenReceive

  \Statex
  \Procedure{\scstart}{$\scsidvar$}
    \label{line:procedure-scstart}
    \If{this is the first operation of the transaction received by the primary}
      \label{line:scstart-if}
      \State $\wtsidvar \gets \rswtconns[\scrsconns[\scsidvar]]$
        \label{line:scstart-wtsidvar}
      \State $\readtsvar \gets \readts[\scsidvar]$
        \label{line:scstart-readtsvar}
      \State $\Call{\wtsetreadts}{\wtsidvar, \readtsvar}$
        \label{line:scstart-call-wtsetreadts}
      \State $\Call{\waitforreadconcern}{\scsidvar}$
        \label{line:scstart-call-waitforreadconcern}
    \EndIf
  \EndProcedure
  \algstore{sc-primary}
  \end{algorithmic}
  \end{varwidth}\qquad
  \begin{varwidth}[t]{0.50\textwidth}
  \begin{algorithmic}[1]
  \algrestore{sc-primary}
  \WhenReceive[$\Call{\prepare}{\scsidvar} \from \primaryvar$]
    \label{line:procedure-prepare}
    \State $\rssidvar \gets \scrsconns[\scsidvar]$
      \label{line:prepare-rssession}
    \State $\ct \gets \Call{tick}{\null}$
      \label{line:prepare-call-tick}
    \State $\Call{\wtprepare}{\rswtconns[\rssidvar], \ct}$
      \label{line:prepare-call-wtprepare}
    \State $\opsvar \gets \txnops[\rssidvar]$
      \label{line:prepare-ops}
    \If{$\opsvar = \emptyset$}
      \label{line:prepare-ops-empty}
      \State $\oplog \gets @ \circ \tuple{\ct, \noop}$
        \label{line:prepare-oplog-noop}
    \Else
      \State $\oplog \gets @ \circ \tuple{\ct, \opsvar}$
        \label{line:prepare-oplog}
    \EndIf
    \State \untilstyle{$\wait\until \lastmajoritycommitted \ge \ct$}
      \label{line:prepare-wait-majoritycommitted}
    \State $\send \Call{\prepareack}{\ct} \sendto \primaryvar$ 
      \label{line:prepare-call-prepareack}
  \EndWhenReceive

  \Statex
  \WhenReceive[$\Call{\commit}{\scsidvar, \committsvar} \from \primaryvar$]
    \label{line:procedure-commit}
    \State $\ct \gets \Call{\tick}{\null}$
      \label{line:commit-ct}
    \State $\wtsidvar \gets \rswtconns[\scrsconns[\scsidvar]]$
      \label{line:commit-wtsession}
    \State $\Call{\wtcommitpreparets}{\wtsidvar, \committsvar}$
      \label{line:commit-call-wtcommitpreparets}
    \State $\Call{\wtcommitprepare}{\wtsidvar}$
      \label{line:commit-call-wtcommitprepare}
    \State $\oplog \gets @ \circ \tuple{\ct, \committsvar}$
      \label{line:commit-oplog}
    \State \untilstyle{$\wait\until \lastmajoritycommitted \ge \ct$}
      \label{line:commit-wait}
    \State $\send \decack(\scsidvar) \sendto \primaryvar$
      \label{line:commit-call-decack}
  \EndWhenReceive

  \Statex
  \WhenReceive[$\Call{\abort}{\scsidvar} \from \mongosvar$]
    \label{line:procedure-abort}
    \State $\Call{\rsrollback}{\scrsconns[\scsidvar]}$
      \label{line:abort-call-rsrollback}
    \State \untilstyle{$\wait\until \lastmajoritycommitted \ge \ct$}
      \label{line:abort-wait}
    \State $\send \decack(\scsidvar) \sendto \mongosvar$
      \label{line:abort-call-decack}
  \EndWhenReceive

  \Statex
  \Procedure{\waitforreadconcern}{$\scsidvar$}
    \label{line:procedure-waitforreadconcern}
    \State $\readtsvar \gets \readts[\scsidvar]$
      \label{line:waitforreadconcern-readts}
    \If{$\ct < \readtsvar$}
      \label{line:waitforreadconcern-ct-if}
      \State $\oplog \gets @ \circ \tuple{\readtsvar, \noop}$
        \label{line:waitforreadconcern-call-noopwrite}
    \EndIf
    \State $\wait\until \Call{\allcommitted}{\null} \ge \readtsvar$
      \label{line:waitforreadconcern-wait-until-allcommitted}
  \EndProcedure
  \end{algorithmic}
  \end{varwidth}
\end{algorithm}


For brevity, we focus on the behavior of the primary nodes in the sharded cluster.
Consider a session $\scsidvar$ connected to a mongos.
We use $\readts[\scsidvar]$ to denote the read timestamp, assigned by the mongos,
of the currently active transaction on the session.
\subsubsection{Read and Update Operations} \label{sss:sc-read-update}

If this is the first operation the primary receives,
it calls \scstart{} to set the transaction's read timestamp in WiredTiger
(line~\code{\ref{alg:sc-primary}}{\ref{line:scstart-call-wtsetreadts}}).
In \scstart{}, it also calls \waitforreadconcern{}
to handle \caseclockskew{} and \caseholes{}
(line~\code{\ref{alg:sc-primary}}{\ref{line:scstart-call-waitforreadconcern}}).

The primary then delegates the operation to \rsalg{}
(lines~\code{\ref{alg:sc-primary}}{\ref{line:scread-call-rsread}}
and~\code{\ref{alg:sc-primary}}{\ref{line:scupdate-call-rsupdate}}).
To handle \casependingcommitread{},
\rsread{} has been modified to keep trying reading from WiredTiger
until it returns a value updated by a committed transaction
(line~\code{\ref{alg:rs-primary}}{\ref{line:rsread-call-wtread-in-while}}).
To handle \casependingcommitupdate{},
\rsupdate{} first performs an \scread{} on the same key
(line~\code{\ref{alg:rs-primary}}{\ref{line:rsupdate-call-rsread}}).
Moreover, if the update fails due to write conflicts,
the mongos will send an $\abort$ message to the primary nodes of all other participants, without entering 2PC.
\subsubsection{Commit Operations} \label{sss:sc-commit}

In 2PC, the transaction coordinator behaves as described in Sections~\ref{sss:sc-2pc}
and~\ref{sss:sc-replication} for atomic commitment and fault tolerance, respectively.
We now explain how the participants handle the \prepare{} and \commit{} messages in more detail.

After receiving a \prepare{} message,
the participant advances its cluster time and takes it as the prepare timestamp
(lines~\code{\ref{alg:sc-primary}}{\ref{line:prepare-call-tick}},
\code{\ref{alg:sc-primary}}{\ref{line:prepare-call-wtprepare}},
\code{\ref{alg:wt-db}}{\ref{line:wtprepare-txn-preparets}},
and \code{\ref{alg:wt-db}}{\ref{line:wtprepare-store}}).
Note that the transaction's $\tid$ in $\wttxnglobal$ is reset to $\wttidnone$
(line~\code{\ref{alg:wt-db}}{\ref{line:wtprepare-wttxnglobal-states}}).
Thus, according to the visibility rule,
this transaction is \emph{visible} to other transactions that starts later in WiredTiger.
Next, the participant creates an oplog entry containing the updates executed locally
or a \noop{} oplog entry for the ``speculative majority'' strategy (Section~\ref{sss:ssi}).
Then, it waits until the oplog entry has been majority committed
(line~\code{\ref{alg:sc-primary}}{\ref{line:prepare-wait-majoritycommitted}}).

When a participant receives a \commit{} message, it ticks its cluster time.
After setting the transaction's commit timestamp
(line~\code{\ref{alg:sc-primary}}{\ref{line:commit-call-wtcommitpreparets}}),
it asks WiredTiger to commit the transaction locally
(line~\code{\ref{alg:sc-primary}}{\ref{line:commit-call-wtcommitprepare}}).
Note that the status of the transaction is changed to \wtprepareresolved{}
(line~\code{\ref{alg:wt-db}}{\ref{line:wtcommitprepare-store}}).
Thus, this transaction is now visible to other waiting transactions
(line~\code{\ref{alg:rs-primary}}{\ref{line:rsread-while-status}}).
Then, the participant generates an oplog entry containing the commit timestamp
and waits for it to be majority committed.



\section{Checking Snapshot Isolation} \label{section:checking-si}

In this section, we design and evaluate
white-box polynomial-time checking algorithms
for the transactional protocols of MongoDB
against \strongsi, \rtsi, and \sessionsi.
The project can be found at \url{https://github.com/Tsunaou/MongoDB-SI-Checker}.


\begin{table}[t]
  \centering
  \caption{Cloud virtual machines used in our experiments.}
  \label{table:machines}
  \renewcommand{\arraystretch}{1.2}
  \resizebox{\textwidth}{!}{%
  \begin{tabular}{|l|l|l|l|}
    \hline
    \textbf{VMs} & \textbf{Configuration} & \textbf{OS} & \textbf{Region} \\ \hline
    VM-a\{0-6\}              & 3.10GHz Intel(R) Xeon(R) Platinum 8269CY CPU with 2 virtual cores and 4GB of RAM                       & Ubuntu 20.04                & Chengdun-A \\ \hline
    VM-b\{0-2\}              & 2.50GHz Intel(R) Xeon(R) Platinum 8269CY CPU with 1 virtual cores and 1GB of RAM                       & Ubuntu 20.04                &  Chengdu-B\\ \hline
    \end{tabular}
  }
\end{table}

\begin{table}[t]
    \centering
    \caption{Configurations of MongoDB deployments.}
    \label{table:deployment}
    \renewcommand{\arraystretch}{1.0}
    \resizebox{\textwidth}{!}{%
    \begin{tabular}{|c|c|c|}
        \hline
        \textbf{Deployment} & \textbf{Version} & \textbf{Configuration} \\ \hline
        Standalone          & WiredTiger 3.3.0 & A standalone WiredTiger storage engine in VM-a0. \\ \hline
        Replica Set         & MongoDB 4.4.5 &
        \begin{tabular}[c]{@{}c@{}}
          A replica set of 5 nodes with VM-a1 as the primary and VM-a\set{2-5} as secondaries.\end{tabular} \\ \hline
        Sharded Cluster     & MongoDB 4.4.5                               &
        \begin{tabular}[c]{@{}c@{}}
          A cluster consisting of 1 config server and 2 shards. \\
          The config server, shard1, and shard2 are all replica sets, and \\
          deployed in VM-a\set{1-3}, VM-a\set{4-6}, and VM-b\set{0-2}, respectively.\end{tabular} \\ \hline
    \end{tabular}
    }
\end{table}


\begin{table}[t]
    \centering
    \caption{Transaction generation parameters (supported by Jepsen).}
    \label{table:generator}
    \renewcommand{\arraystretch}{1.3}
    \resizebox{\textwidth}{!}{%
        \begin{tabular}{|llllllll|}
            \hline
            \multicolumn{8}{|c|}{\textbf{Tunable Parameters}}                                                                                                                                                                                                                                                                   \\ \hline
            \multicolumn{1}{|l|}{\textbf{Parameters}}        & \multicolumn{1}{l|}{\textbf{Default}}  & \multicolumn{5}{l|}{\textbf{Range}}                    & \multicolumn{1}{l|}{\textbf{Description}} \\ \hline
            \multicolumn{1}{|l|}{\param{txn-num}}            & \multicolumn{1}{l|}{3000}              & \multicolumn{5}{l|}{\set{1000, 2000, 3000, 4000, 5000}} & \multicolumn{1}{l|}{The total number of transactions.} \\ \hline
            \multicolumn{1}{|l|}{\param{concurrency}}        & \multicolumn{1}{l|}{9}                 & \multicolumn{5}{l|}{\set{3, 6, 9, 12, 15}}              & \multicolumn{1}{l|}{The number of clients.}\\ \hline
            \multicolumn{1}{|l|}{\param{max-txn-len}}        & \multicolumn{1}{l|}{12}                & \multicolumn{5}{l|}{\set{4, 8, 12, 16, 20}}             & \multicolumn{1}{l|}{The maximum number of operations in each transaction.}\\ \hline
            \hline
            \multicolumn{8}{|c|}{\textbf{Fixed Parameters}}                                                                                                                                                                                                                                                                          \\ \hline
            \multicolumn{1}{|l|}{\textbf{Parameters}}         & \multicolumn{3}{l|}{\textbf{Value}}    & \multicolumn{4}{l|}{\textbf{Description}}                                                                                                                                                                \\ \hline
            \multicolumn{1}{|l|}{\param{key-count}}          & \multicolumn{3}{l|}{10}                & \multicolumn{4}{l|}{There are 10 distinct keys at any point for generation.}                                                                                                                             \\ \hline
            \multicolumn{1}{|l|}{\param{max-writes-per-key}} & \multicolumn{3}{l|}{128}               & \multicolumn{4}{l|}{There are at most 128 updates on each key.}                                                                                                                                          \\ \hline
            \multicolumn{1}{|l|}{\param{key-dist}}           & \multicolumn{3}{l|}{exponential}       & \multicolumn{4}{l|}{Probability distribution for keys.} \\ \hline
            \multicolumn{1}{|l|}{\param{read:update ratio}}  & \multicolumn{3}{l|}{1 : 1}             & \multicolumn{4}{l|}{The default (and fixed) read:update ration in Jepsen.}                                                                                                                               \\ \hline
            \multicolumn{1}{|l|}{\param{timeout}}            & \multicolumn{3}{l|}{5s/10s/30s}        & \multicolumn{4}{l|}{The timeout for \wtalg/\rsalg/\scalg{} transactions.} \\ \hline
        \end{tabular}
    }
\end{table}

\subsection{White-box Checking Algorithms} \label{ss:checking-algs}

\subsubsection{Basic Ideas} \label{sss:checking-ideas}

The three checking algorithms work in the same manner.
For example, in the checking algorithm for the \wtalg{} protocol against \strongsi,
we first extract the $\viswt$ and $\arwt$ relations from a given history $\h$
according to Definitions~\ref{def:viswt} and \ref{def:arwt}.
Then, we check whether the abstract execution $\ae \triangleq (\h, \viswt, \arwt)$
satisfies all the axioms required by $\strongsi$ according to Definition~\ref{def:strongsi}.
Since the total order $\arwt$ offers the version order~\cite{Adya:PhDThesis1999, OOPSLA2019:Complexity},
this checking, particularly for $\extaxiom$, can be easily done in polynomial time.
\subsubsection{Implementation Considerations} \label{sss:impl}

We now explain how to obtain the additional information necessary for
extracting the appropriate $\vis$ and $\ar$ relations for each protocol.

For each \wtalg{} transaction $\txnvar$, we record the real time when it starts and commits.
They are taken as $\starttimewt(\txnvar)$ and $\committimewt(\txnvar)$, respectively.
However, this poses a technical challenge involving real time:
Due to the potential non-accurate records of physical time,
it is possible that a transaction $\txnvar'$ reads data written by
another transaction $\txnvar$, but $\txnvar'$ starts before $\txnvar$ commits, violating the $\inrbaxiom$ axiom.
Therefore, we measure the \emph{degree of inaccuracy}:
we enumerate all pairs of transactions like $\txnvar'$ and $\txnvar$ above
and take the maximum over all $\committimewt(\txnvar) - \starttimewt(\txnvar')$.
We call it the ``real time error'' and report it in the experiments.
To further reduce the impacts caused by the real time issue,
we can also utilize some other information to fix $\viswt$ and $\arwt$.
For example, Lemma~\ref{lemma:wt-conflict-arwt-tid} says that for any two conflicting transactions,
the $\arwt$ relation between them can be defined by their transaction identifiers.
Therefore, in the experiments, we also obtain the transaction identifiers of update transactions
from the write-ahead log of WiredTiger,
and omitted the histories which violate Lemma~\ref{lemma:wt-conflict-arwt-tid}.

For \rsalg, we obtain the read timestamps of transactions from the \texttt{mongod.log} file,
and commit timestamps from the oplog (stored in the collection \texttt{oplog.rs}) on the primary.

For \scalg, we obtain the transactions' read timestamps from the \texttt{mongod.log} file.
For update transactions that involve only a single shard,
their commit timestamps are recorded in the oplog of the primary of the shard.
For a distributed transaction across multiple shards that enters 2PC,
its commit timestamp is stored in the oplog entry on the primary of any shard involved.
\subsection{Evaluations} \label{ss:evaluation}
\subsubsection{Experimental Setup} \label{sss:experiment-setup}

To demonstrate the effectiveness and the efficiency of the white-box checking algorithms,
we implement them and use them to test MongoDB in different deployments.

We utilize the Jepsen testing framework to schedule the execution of transactions in each deployment.
Each MongoDB deployment is considered as a key-value store as described in Section~\ref{section:si}.
A group of clients concurrently issue random transactions to MongoDB.
Table~\ref{table:generator} summaries both the tunable and fixed parameters for executions.
We consider 10 histories for each combination of parameter values.
Note that we use relatively long timeout values
to reduce possible rollbacks triggered by client timeout,
since we do not consider failures in the work.

Our experiments evaluate WiredTiger 3.3.0 and MongoDB 4.4.5.
Tables~\ref{table:machines} and \ref{table:deployment} shows the configurations of the machines and the deployments.
The average RTT in Region-A, Region-B, and between regions
are 0.14ms, 0.19ms, and 0.86ms, respectively.
\subsubsection{Experimental Results} \label{sss:experiment-result}

First of all, the experimental results confirm our theoretical analysis:
the transactional protocols on the replica set deployment
and the sharded cluster deployment
satisfy \rtsi{} and \sessionsi, respectively.
The transactional protocol in WiredTiger satisfies \strongsi{},
if we can tolerate the real time error within about $25$ms,
which is crucial for validating the $\rbaxiom$ axiom.


\begin{figure}[t]
  \centering
  \includegraphics[width = 1.00\textwidth]{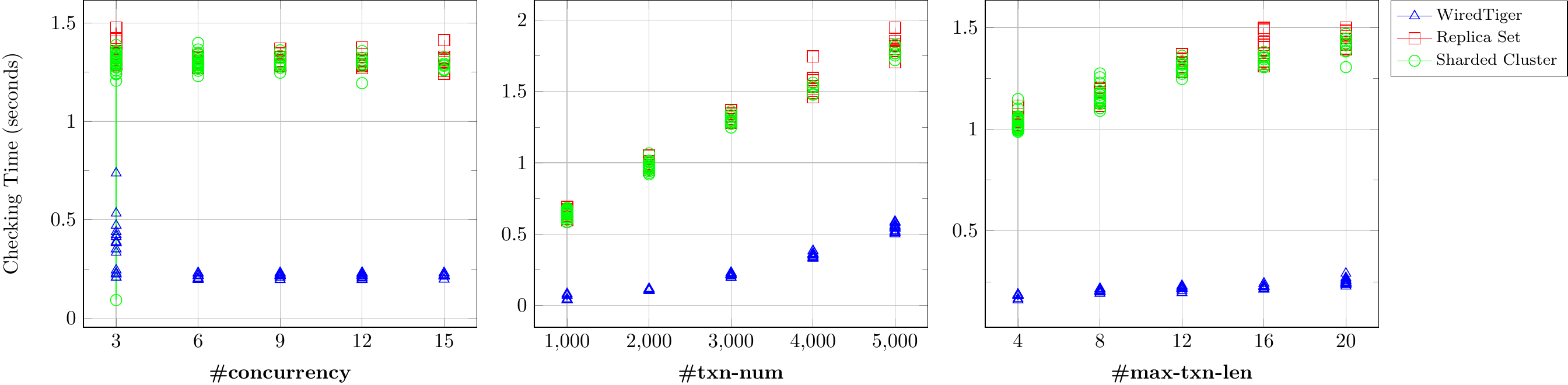}
  \caption{Checking time against different snapshot isolation variants on three MongoDB deployments.}
  \label{fig:si-check-perf}
\end{figure}

Figure~\ref{fig:si-check-perf} shows the performance of our white-box checking algorithms.
All experiments are performed on VM-a0.
The three algorithms are fast to check the transactions for all deployments,
since each consistency axiom can be efficiently checked.
For example, it takes less than 2s to check against a replica set or sharded cluster history
consisting of 5000 transactions.
Therefore, they are quite efficient for checking the satisfaction of
large-scale MongoDB transactional workloads against variants of snapshot isolation.



\section{Related Work} \label{section:related-work}

\emph{Consistency Models in MongoDB.}
Schultz \emph{et al.}~\cite{VLDB2019:TunableConsistency} discussed
the tunable consistency models in MongoDB,
which allow users to select the trade-off between consistency and latency
at a per operation level by choosing different
\rc, \wc, and \emph{readPreference} parameters.
Tyulenev \emph{et al.}~\cite{SIGMOD2019:MongoDB-CC} discussed
the design and implementation of causal consistency introduced by MongoDB 3.6,
which provides session guarantees including read-your-writes, monotonic-reads, monotonic-writes,
and writes-follow-reads~\cite{SessionGuarantee:PDIS1994, SessionCausal:PDNP04}.
In this work, we are concerned about the more challenging transactional consistency of MongoDB.

\emph{Specification and Verification of MongoDB.}
The MongoDB Inc. has been actively working on the
formal specification and verification of MongoDB.
Zhou \emph{et al.}~\cite{NSDI2021:MongoDB-Raft}
presented the design and implementation of MongoDB Raft,
the pull-based fault-tolerant replication protocol in MongoDB.
Schultz \emph{et al.}~\cite{MongoDBReconfig:DISC2021}
presented a novel logless dynamic reconfiguration protocol for the MongoDB replication system.
Both protocols have been formally specified
using TLA$^+$~\cite{SpecifyingSystems:Lamport2002}
and verified using the TLC model checker~\cite{Yu:CHARME99}.
In this work, we formally specify and verify
the transactional consistency protocols of MongoDB.

\emph{Transactional Consistency Checking.}
Biswas \emph{et al.}~\cite{OOPSLA2019:Complexity} have studied the complexity issues
of checking whether a given history of a database
satisfies some transactional consistency model.
Particularly, they proved that it is \textsf{NP-complete}
to check the satisfaction of (Adya) snapshot isolation
for a history without the version order.
Gan \emph{et al.}~\cite{VLDB2020:IsoDiff} proposed IsoDiff,
a tool for debugging anomalies caused by weak isolation for an application.
IsoDiff finds a representative subset of anomalies.
Kingsbury \emph{et al.}~\cite{VLDB2020:Elle} presented Elle,
a checker which infers an Adya-style dependency graph
among client-observed transactions.
By carefully choosing appropriate data types and operations,
Elle is able to infer the version order of a history.
Tan \emph{et al.}~\cite{Cobra:OSDI2020} proposed Cobra,
a black-box checking system of serializability.
By leveraging the SMT solver MonoSAT~\cite{SMT:AAAI2015} and other optimizations,
Cobra can handle real-world online transactional processing workloads.
In this work, we design white-box polynomial-time checking algorithms for
the transactional protocols of MongoDB against \strongsi, \rtsi, and \sessionsi.
These checking algorithms make use of the properties of the transactional protocols
to infer the version order of histories,
effectively circumventing the \textsf{NP-hard} obstacle in theory~\cite{OOPSLA2019:Complexity}.



\section{Conclusion and Future Work} \label{section:conclusion}

In this paper, we have formally specified and verified
the transactional consistency protocols in each MongoDB deployment,
namely \wtalg, \rsalg, and \scalg{}.
We proved that they satisfy different variants of snapshot isolation,
namely \strongsi, \rtsi, and \sessionsi, respectively.
We have also proposed and evaluated efficient white-box checking algorithms
for MongoDB transaction protocols against their consistency guarantees.

Our work is a step towards formally verifying MongoDB,
and creates plenty of opportunities for future research.
We list some possible future work.
\begin{itemize}
  \item First, for simplicity, we have assumed that each procedure executes atomically.
	However, the implementation of MongoDB is highly concurrent
	with intricate locking mechanisms~\cite{VLDB2020:eXtreme},
	and needs to be modelled and verified more carefully.
  \item Second, we did not consider failures in this paper.
	MongoDB employs a pull-based Raft protocol for data replication
	and tolerates any minority of node failures~\cite{NSDI2021:MongoDB-Raft}.
	We plan to verify and evaluate MongoDB Raft in the future work.
	We will also explore the fault-tolerance and recovery mechanisms in distributed transactions.
  \item Third, there is a gap between our (simplified) pseudocode and the real implementation of MongoDB.
    We plan to use model-based testing~\cite{VLDB2020:eXtreme} to bridge such a gap.
  \item Fourth, MongoDB also supports non-transactional consistency,
	including tunable consistency~\cite{VLDB2019:TunableConsistency}
	and causal consistency~\cite{SIGMOD2019:MongoDB-CC}.
	It is unclear what the consistency model is like
	when non-transactional operations are involved.
  \item Finally, it is interesting to explore the transactional consistency checking algorithms
	utilizing SMT solvers~\cite{OOPSLA2019:Complexity, Cobra:OSDI2020}.
\end{itemize}

\bibliography{mongodb-tx}
\appendix

\section{More Variants of Snapshot Isolation} \label{section:app-si}

Generalized snapshot isolation (denoted \gsi{}),
which is shown equivalent to \ansisi{} in~\cite{PODC2017:ClientCentric},
limits a transaction $T$ to read only from snapshots
that do \emph{not} include transactions that committed in real time \emph{after}
$T$ starts (i.e., \inrbaxiom), and
requires the total order $\ar$ respect the commit order (i.e., \cbaxiom).

\begin{definition} \label{def:gsi}
  \[
    \gsi = \si \land \inrbaxiom \land \cbaxiom.
  \]
\end{definition}



\section{Correctness Proofs} \label{section:app-proofs}


\subsection{Correctness of \wtalg} \label{ss:wt-correctness}

Consider a history $\h = (\WTTXN, \sowt)$ of \wtalg{},
where we restrict $\WTTXN$ to be set of committed transactions.
Denote by $\WTTXNUPDATE \subseteq \WTTXN$ be the set of update transactions in $\WTTXN$.
We prove that $\h$ satisfies \strongsi{}
by constructing an abstract execution $\ae = (\h, \viswt, \arwt)$
(Theorem~\ref{thm:wtalg-strongsi}).
\subsubsection{The Visibility Relation} \label{sss:vis-wt}

\begin{definition}[$\viswt$] \label{def:viswt}
  $\viswt \triangleq \rbwt.$
\end{definition}

\begin{lemma} \label{lemma:wt-realtimesnapshotaxiom}
  $\ae \models \realtimesnapshotaxiom.$
\end{lemma}

\begin{lemma} \label{lemma:viswt-tid}
  $\forall \txnvar, \txnvar' \in \WTTXNUPDATE.\;
    \txnvar \rel{\viswt} \txnvar' \implies \txnvar.\tid < \txnvar'.\tid.$
\end{lemma}

\begin{proof} \label{proof:viswt-tid}
  Consider any two transactions $\txnvar, \txnvar' \in \WTTXN$
  that have obtained non-zero transaction identifiers
  such that $\txnvar \rel{\viswt} \txnvar'$.
  By Definition~\ref{def:viswt} of $\viswt$, $\txnvar \rel{\rbwt} \txnvar'.$
  By lines~\code{\ref{alg:wt-db}}{\ref{line:wtupdate-txn-tid}}
  and \code{\ref{alg:wt-db}}{\ref{line:wtupdate-wttxnglobal-currentid}},
  $\txnvar.\tid < \txnvar'.\tid.$
\end{proof}

\begin{definition}[$\wtvis$] \label{def:wtvis}
  For a transaction $\txnvar \in \WTTXN$,
  we define $\wtvis(\txnvar) \subseteq \WTTXNUPDATE$
  to be the set of update transactions (excluding $\txnvar$ itself)
  that are visible to $\txnvar$ according to the visibility rule in $\txnvis$.
\end{definition}

\begin{lemma} \label{lemma:wtvis-welldefined}
  For any transaction $\txnvar \in \WTTXN$, $\wtvis(\txnvar)$ is well-defined.
  That is, $\wtvis(\txnvar)$ does not change over the lifecycle of $\txnvar$.
\end{lemma}

\begin{proof} \label{proof:wtvis-welldefined}
  Consider a transaction $\txnvar \in \WTTXN$ which starts at time $\timepoint$.
  Denote by $\wtvis_{\timepoint}(\txnvar)$ the set of transactions that are visible to $\txnvar$ at time $\timepoint$.
  Let $\txnvar' \neq \txnvar \in \WTTXNUPDATE$ be an update transaction on session $\wtsidvar$.
  Suppose that $\txnvar'$ obtains its transaction identifier $\txnvar'.\tid$ at time $\timepoint'$
  and finishes at time $\timepoint''$.
  We distinguish between two cases.

  \casei: $\timepoint < \timepoint''$. There are two cases.
    If $\timepoint' < \timepoint$, we have $\wttxnglobal[\wtsid] \neq \wttidnone$ at time $\timepoint$.
      Therefore, $\txnvar'.\tid \in \txnvar.\concur$.
      Since $\txnvar.\concur$ does not change, $\txnvar'$ is invisible to $\txnvar$ over the lifecycle of $\txnvar$.
    If $\timepoint < \timepoint'$, then $\txnvar'.\tid \ge \txnvar.\upperlimit$.
      Since $\txnvar.\upperlimit$ does not change, $\txnvar'$ is invisible to $\txnvar$ over the lifecycle of $\txnvar$.

  \caseii: $\timepoint'' < \timepoint$.
    Therefore, $\txnvar'.\tid \notin \txnvar.\concur \land \txnvar'.\tid < \txnvar.\upperlimit$.
    Since either $\txnvar.\concur$ or $\txnvar.\upperlimit$ does not change,
    $\txnvar'$ is visible to $\txnvar$ over the lifecycle of $\txnvar$.
\end{proof}

By Lemma~\ref{lemma:wtvis-welldefined},
we can extend the definition of $\wtvis$ to events.

\begin{definition}[$\wtvis$ for Events] \label{def:wtvis-events}
  Let $\evar$ be an event of transaction $\txnvar$.
  We abuse the notation $\wtvis(\evar)$ to denote the set of update transactions
  that are visible to $\evar$ when it occurs.
  We have $\wtvis(\evar) = \wtvis(\txnvar)$.
\end{definition}
\begin{lemma} \label{lemma:viswt-wtvis}
  $\forall \txnvar \in \WTTXN.\; \viswt^{-1}(\txnvar) \cap \WTTXNUPDATE = \wtvis(\txnvar).$
\end{lemma}

\begin{proof} \label{proof:viswt-wtvis}
  We need to show $\rbwt^{-1}(\txnvar) \cap \WTTXNUPDATE = \wtvis(\txnvar).$

  We first show that $\wtvis(\txnvar) \subseteq \rbwt^{-1}(\txnvar) \cap \WTTXNUPDATE$.
  Consider an update transaction $\txnvar' \in \wtvis(\txnvar)$.
  We need to show that $\txnvar' \rel{\rbwt} \txnvar$.
  Suppose by contradiction that $\txnvar$ starts before $\txnvar'$ commits.
  By the same argument in (\casei) of the proof of Lemma~\ref{lemma:wtvis-welldefined},
  $\txnvar' \notin \wtvis(\txnvar)$. Contradiction.

  Next we show that $\rbwt^{-1}(\txnvar) \cap \WTTXNUPDATE \subseteq \wtvis(\txnvar)$.
  Consider an update transaction $\txnvar' \in \rbwt^{-1}(\txnvar)$.
  By the same argument in (\caseii) of the proof of Lemma~\ref{lemma:wtvis-welldefined},
  $\txnvar' \in \wtvis(\txnvar)$.
\end{proof}

\begin{lemma} \label{lemma:wt-noconflictaxiom}
  $\ae \models \noconflictaxiom.$
\end{lemma}

\begin{proof} \label{proof:wt-nonconflictaxiom}
  Consider any two update transactions $\txnvar, \txnvar' \in \WTTXNUPDATE$
  such that $\txnvar \conflict \txnvar'$.
  Without loss of generality, assume that they both update key $\keyvar$,
  and that $\txnvar$ updates $\keyvar$ before $\txnvar'$ does.
  By Definition~\ref{def:viswt} of $\viswt$, $\lnot(\txnvar' \rel{\viswt} \txnvar)$.
  We show that $\txnvar \rel{\viswt} \txnvar'$.
  Suppose by contradiction that $\lnot(\txnvar \rel{\viswt} \txnvar')$.
  By Lemma~\ref{lemma:viswt-wtvis}, $\txnvar \notin \wtvis(\txnvar')$.
  Therefore, $\txnvar'$ would abort when it updates $\keyvar$.
\end{proof}
\subsubsection{The Arbitrary Relation} \label{sss:ar-wt}

\begin{definition}[$\arwt$] \label{def:arwt}
  $\arwt \triangleq \cbwt.$
\end{definition}

\begin{lemma} \label{lemma:wt-cbaxiom}
  $\ae \models \cbaxiom.$
\end{lemma}


\begin{lemma} \label{lemma:viswt-arwt}
  $\viswt \subseteq \arwt.$
\end{lemma}

\begin{proof} \label{proof:viswt-arwt}
  By Definition~\ref{def:viswt} of $\viswt$ and Definition~\ref{def:arwt} of $\arwt$,
  $\viswt = \rbwt \subseteq \cbwt = \arwt$.
\end{proof}

\begin{lemma} \label{lemma:wt-prefixaxiom}
  $\ae \models \prefixaxiom.$
\end{lemma}

\begin{proof} \label{proof:wt-prefixaxiom}
  By Definition~\ref{def:viswt} of $\viswt$, Definition~\ref{def:arwt} of $\arwt$,
  Definition~\ref{def:rb} of $\rb$, and Definition~\ref{def:cb} of $\cb$,
  $\arwt \comp \viswt = \cbwt \comp \rbwt \subseteq \rbwt = \viswt$.
\end{proof}

\begin{lemma} \label{lemma:wt-conflict-arwt-tid}
  $\forall \txnvar, \txnvar' \in \WTTXN.\;
    \txnvar \conflict \txnvar' \implies (\txnvar \rel{\arwt} \txnvar'
      \iff \txnvar.\tid < \txnvar'.\tid).$
\end{lemma}
\begin{proof} \label{proof:wt-conflict-arwt-tid}
  Consider any two transactions $\txnvar, \txnvar' \in \WTTXN$
  such that $\txnvar \conflict \txnvar'$.
  By Lemma~\ref{lemma:wt-noconflictaxiom},
  $\txnvar \rel{\viswt} \txnvar' \lor \txnvar' \rel{\viswt} \txnvar.$
  In the following, we proceed in two directions.

  Suppose that $\txnvar \rel{\arwt} \txnvar'$.
  By Lemma~\ref{lemma:viswt-arwt}, $\lnot(\txnvar' \rel{\viswt} \txnvar)$.
  Therefore, $\txnvar \rel{\viswt} \txnvar'$.
  By Lemma~\ref{lemma:viswt-tid}, $\txnvar.\tid < \txnvar'.\tid$.

  Suppose that $\txnvar.\tid < \txnvar'.\tid$.
  By Lemma~\ref{lemma:viswt-tid}, it cannot be that $\lnot(\txnvar' \rel{\viswt} \txnvar)$.
  Therefore, $\txnvar \rel{\viswt} \txnvar'$.
  By Lemma~\ref{lemma:viswt-arwt}, $\txnvar \rel{\arwt} \txnvar'$.
\end{proof}

\begin{lemma} \label{lemma:wt-intaxiom}
  $\ae \models \intaxiom.$
\end{lemma}

\begin{proof} \label{proof:wt-intaxiom}
  Consider a transaction $\txnvar \in \WTTXN$.
  Let $\evar$ be an event such that $\op(\evar) = \readevent(\keyvar, \valvar)$
  is an \emph{internal} $\wtread$ operation in $\txnvar$.
  Let $\evar' \triangleq \max_{\po}\set{\fvar \mid \op(\fvar) = \_(\keyvar, \_) \land \fvar \rel{\po} e}$.
  We need to show that $\evar' = \_(\keyvar, \valvar)$.

  By Lemma~\ref{lemma:wtvis-welldefined}, $\wtvis(e') = \wtvis(e)$.
  If $\op(\evar') = \readevent(\keyvar, \_)$, $\evar$ obtains the same value as $\evar'$.
  Therefore, $\op(\evar') = \readevent(\keyvar, \valvar)$.
  If $\op(\evar') = \writeevent(\keyvar, \_)$, by the visibility rule,
  $\evar'$ is visible when $\evar$ occurs.
  Thus, $\evar$ reads from $\evar'$.
  Therefore, $\op(\evar') = \writeevent(\keyvar, \valvar)$.
\end{proof}

\begin{lemma} \label{lemma:wt-extaxiom}
  $\ae \models \extaxiom.$
\end{lemma}

\begin{proof} \label{proof:wt-extaxiom}
  Consider a transaction $\txnvar \in \WTTXN$.
  Let $\evar$ be an event such that $\op(\evar) = \readevent(\keyvar, \valvar)$
  is an \emph{external} $\wtread$ operation in $\txnvar$.
  Let $W \triangleq \viswt^{-1}(\txnvar) \cap \WriteTx_{\keyvar}$
  be the set of transactions that update key $\keyvar$ and are visible to $\txnvar$.
  If $W = \emptyset$, obviously $e$ obtains the initial value of $\keyvar$.
  Now suppose that $W \neq \emptyset$.
  By Lemma~\ref{lemma:viswt-wtvis}, $W = \wtvis(\txnvar) \cap \WriteTx_{\keyvar}$.
  By Lemma~\ref{lemma:wt-conflict-arwt-tid},
  the $\arwt|_{W}$ order is consistent with the increasing $\tid$ order of the transactions in $W$,
  which is also the list order at line~\code{\ref{alg:wt-db}}{\ref{line:wtread-call-txnvis}}.
  Therefore, $\evar$ reads from $\max_{\arwt} W$.
  Thus, $\max_{\arwt} W \vdash \writeevent(\keyvar, \valvar)$.
\end{proof}

\begin{theorem} \label{thm:wtalg-strongsi}
  $\wtalg \models \strongsi.$
\end{theorem}

\begin{proof} \label{proof:wtalg-strongsi}
  For any history $\h$ of $\wtalg$,
  we construct an abstract execution $\ae = (\h, \viswt, \arwt)$,
  where $\viswt$ and $\arwt$ are given in Definitions~\ref{def:viswt} and \ref{def:arwt}, respectively.
  By Lemmas~\ref{lemma:wt-realtimesnapshotaxiom},
  \ref{lemma:wt-noconflictaxiom}, \ref{lemma:wt-cbaxiom}, \ref{lemma:wt-prefixaxiom},
  \ref{lemma:wt-intaxiom}, and \ref{lemma:wt-extaxiom}, $\ae \models \strongsi$.
  Since $\h$ is arbitrary, $\wtalg \models \strongsi$.
\end{proof}


\subsection{Correctness of \rsalg} \label{ss:rs-correctness}

Consider a history $\h = (\RSTXN, \sors)$ of \rsalg{},
where we restrict $\RSTXN$ to be set of committed transactions.
Denote by $\RSTXNUPDATE \subseteq \RSTXN$ be the set of update transactions in $\RSTXN$.
We prove that $\h$ satisfies \rtsi{} by constructing an abstract execution $\ae = (\h, \visrs, \arrs)$
(Theorem~\ref{thm:rsalg-rtsi}).

\begin{lemma} \label{lemma:rs-majority-commit}
  Transactions are majority committed in the increasing order of their commit timestamps.
\end{lemma}

\begin{proof} \label{proof:rs-majority-commit}
  By the replication mechanism (Algorithm~\ref{alg:replication}).
\end{proof}

\begin{definition}[$\visrs$] \label{def:visrs}
  $\forall \txnvar, \txnvar' \in \RSTXN.\;
    \txnvar \rel{\visrs} \txnvar' \iff
      \txnvar.\committs \le \txnvar'.\readts.$
\end{definition}

\begin{lemma} \label{lemma:visrs-acyclic}
  $\visrs$ is acyclic.
\end{lemma}

\begin{proof} \label{proof:visrs-acyclic}
  This holds due to the property that
  $\forall \txnvar \in \RSTXN.\; \txnvar.\readts < \txnvar.\committs$
  (line~\code{\ref{alg:rs-primary}}{\ref{line:rscommit-tick}}).
\end{proof}

\begin{lemma} \label{lemma:rs-rbaxiom}
  $\ae \models \rbaxiom.$
\end{lemma}

\begin{proof} \label{proof:rs-rbaxiom}
  Consider any two transactions $\txnvar, \txnvar' \in \RSTXN$
  such that $\txnvar \rel{\rbrs} \txnvar'$.
  By line~\code{\ref{alg:rs-primary}}{\ref{line:rscommit-wait-majoritycommitted}},
  when $\txnvar'$ starts, $\txnvar$ has been majority committed.
  By Lemma~\ref{lemma:rs-majority-commit},
  all transactions with commit timestamps $\le \txnvar.\committs$
  have been majority committed and thus locally committed on the primary.
  Therefore, when $\txnvar'.\readts$ is computed
  (line~\code{\ref{alg:rs-primary}}{\ref{line:openwtsession-call-allcommitted}}),
  all $\pinnedcommittedtsvar \neq \wttsnone$ in $\wttxnglobal$ are larger than $\txnvar.\committs$
  (line~\code{\ref{alg:wt-db}}{\ref{line:allcommitted-return}}).
  Thus, $\txnvar.\committs \le \txnvar'.\readts$.
  By Definition~\ref{def:visrs} of $\visrs$, $\txnvar \rel{\visrs} \txnvar'$.
\end{proof}

Since $\rsalg$ transactions are encapsulated into $\wtalg$ transactions,
we can extend $\viswt$ over $\rsalg$ transactions.
Generally speaking, the following lemma shows that
the timestamps of $\rsalg$ overrides the transaction identifiers of $\wtalg$.

\begin{lemma} \label{lemma:visrs-viswt}
  $\visrs \subseteq \viswt.$
\end{lemma}

\begin{proof} \label{proof:visrs-viswt}
  Consider two transactions $\txnvar$ and $\txnvar'$ in $\RSTXN$ such that $\txnvar \rel{\visrs} \txnvar'$.
  By Definition~\ref{def:visrs} of $\visrs$, $\txnvar.\committs \le \txnvar'.\readts$.
  We need to show that $\txnvar \rel{\viswt} \txnvar'$,
  which, by Definition~\ref{def:viswt} of $\viswt$, is $\committimewt(\txnvar) < \starttimewt(\txnvar')$.

  Consider the following three time points:
  $\timepoint_{1}$ when $\txnvar$ obtained its commit timestamp $\txnvar.\committs$
  (line~\code{\ref{alg:rs-primary}}{\ref{line:rscommit-txn-committs}})
  and atomically set $\txnvar.\committs$ in WiredTiger
  (line~\code{\ref{alg:wt-db}}{\ref{line:wtsetcommitts-wttxnglobal-states}}),
  $\timepoint_{2} \triangleq \committimewt(\txnvar)$
  when $\txnvar$ was locally committed in WiredTiger on the primary
  (line~\code{\ref{alg:wt-db}}{\ref{line:procedure-wtcommit}}),
  and $\timepoint_{3} \triangleq \starttimewt(\txnvar')$
  when $\txnvar'$ computed its read timestamp $\txnvar'.\readts$
  (lines~\code{\ref{alg:rs-primary}}{\ref{line:openwtsession-call-allcommitted}}
  and~\code{\ref{alg:wt-db}}{\ref{line:procedure-allcommitted}}).
  We need to show that $\timepoint_{2} < \timepoint_{3}$.
  Note that $\timepoint_{1} < \timepoint_{2}$.

  We first show that $\timepoint_{1} < \timepoint_{3}$.
  Suppose by contradiction that $\timepoint_{3} < \timepoint_{1}$.
  Since the commit timestamps of $\rsalg$ transactions are strictly monotonically increasing,
  by the way $\txnvar'.\readts$ is computed,
  $\txnvar'.\readts < \txnvar.\committs$. Contradiction.

  Next we show that $\timepoint_{2} < \timepoint_{3}$.
  Suppose by contradiction that $\timepoint_{3} < \timepoint_{2}$.
  Since $\timepoint_{1} < \timepoint_{3}$,
  $\timepoint_{1} < \timepoint_{3} < \timepoint_{2}$.
  That is, at time $\timepoint_{3}$,
  $\txnvar$ has obtained its commit timestamp, but has not been locally committed.
  By the way $\txnvar'.\readts$ is computed,
  $\txnvar'.\readts < \txnvar.\committs$. Contradiction.
\end{proof}

\begin{definition} \label{def:rsvis}
  For a transaction $\txnvar \in \RSTXN$,
  we define $\rsvis(\txnvar) \subseteq \RSTXNUPDATE$
  to be the set of update transactions in $\RSTXN$
  that are visible to $\txnvar$ according to the visibility rule in $\txnvis$.
\end{definition}

\begin{lemma} \label{lemma:rsvis-welldefined}
  For any transaction $\txnvar \in \RSTXN$,
  $\rsvis(\txnvar)$ is well-defined.
  That is, $\rsvis(\txnvar)$ does not change over the lifecycle of $\txnvar$.
\end{lemma}

\begin{proof} \label{proof:rsvis-welldefined}
  By Lemma~\ref{lemma:wtvis-welldefined} and the fact that the read timestamp
  and commit timestamp of a transaction do not change over its lifecycle.
\end{proof}

As with $\wtvis$, by Lemma~\ref{lemma:rsvis-welldefined},
we can extend the definition of $\rsvis$ to events.

\begin{definition}[$\rsvis$ for Events] \label{def:rsvis-events}
  Let $\evar$ be an event of transaction $\txnvar$.
  We abuse the notation $\rsvis(\evar)$ to denote the set of update transactions
  that are visible to $\evar$ when it occurs.
  We have $\rsvis(\evar) = \rsvis(\txnvar)$.
\end{definition}

\begin{lemma} \label{lemma:visrs-rsvis}
  $\forall \txnvar \in \RSTXN.\; \visrs^{-1}(\txnvar) \cap \RSTXNUPDATE = \rsvis(\txnvar).$
\end{lemma}

\begin{proof} \label{proof:visrs-rsvis}
  Consider a transaction $\txnvar \in \RSTXN$.
  By Definition~\ref{def:visrs} of $\visrs$,
  $\visrs^{-1}(\txnvar) \cap \RSTXNUPDATE =
    \set{\txnvar' \in \RSTXNUPDATE \mid \txnvar'.\committs \le \txnvar.\readts}$.
  By the visibility rule (line~\code{\ref{alg:wt-db}}{\ref{line:procedure-txnvis}})
  and Lemmas~\ref{lemma:visrs-viswt} and \ref{lemma:rsvis-welldefined},
  $\rsvis(\txnvar) =
    \set{\txnvar' \in \RSTXNUPDATE \mid \txnvar'.\committs \le \txnvar.\readts}$.
  Thus, $\visrs^{-1}(\txnvar) \cap \RSTXNUPDATE = \rsvis(\txnvar)$.
\end{proof}

\begin{lemma} \label{lemma:rs-noconflictaxiom}
  $\ae \models \noconflictaxiom.$
\end{lemma}

\begin{proof} \label{proof:rs-noconflictaxiom}
  Consider any two transactions $\txnvar, \txnvar' \in \RSTXN$ such that $\txnvar \conflict \txnvar'$.
  Suppose that they both update key $\keyvar$.
  Suppose by contradiction that
  $\lnot (\txnvar \rel{\visrs} \txnvar' \lor \txnvar' \rel{\visrs} \txnvar)$.
  By Definition~\ref{def:visrs} of $\visrs$, that is
  $(\txnvar.\readts < \txnvar'.\committs) \land (\txnvar'.\readts < \txnvar.\committs)$.

  By Lemma~\ref{lemma:wt-noconflictaxiom},
  $\txnvar$'s $\wtcommit$ finishes before $\txnvar'$'s $\wtstart$ starts
  \emph{or} $\txnvar'$'s $\wtcommit$ finishes before $\txnvar$'s $\wtstart$ starts.
  These two cases are symmetric.
  In the following, we consider the first case
  which implies that $\txnvar$ updates $\keyvar$ before $\txnvar'$ does.
  When $\txnvar'$ updates $\keyvar$,
  $\tuple{\txnvar.\tid, \_, \txnvar.\committs} \in \store[\kvar]$
  (line~\code{\ref{alg:wt-db}}{\ref{line:wtupdate-forall}}).
  However, the check $\txnvis(\txnvar', \_, \txnvar.\committs)$
  (line~\code{\ref{alg:wt-db}}{\ref{line:wtupdate-call-txnvis}}) fails
  because $\txnvar'.\readts < \txnvar.\committs$.
  Therefore, $\txnvar'$ would abort
  (line~\code{\ref{alg:wt-db}}{\ref{line:wtupdate-call-rollback}}).
\end{proof}
\subsubsection{The Arbitrary Relation} \label{sss:rs-ar}

\begin{definition}[$\arrs$] \label{def:arrs}
  $\forall \txnvar_{1}, \txnvar_{2} \in \RSTXN.\;
    \txnvar_{1} \rel{\arrs} \txnvar_{2} \iff
      \txnvar_{1}.\committs < \txnvar_{2}.\committs.$
\end{definition}

\begin{lemma} \label{lemma:arrs-is-total}
  $\arrs$ is a total order.
\end{lemma}

\begin{proof} \label{proof:arrs-is-total}
  By line~\code{\ref{alg:rs-primary}}{\ref{line:rscommit-tick}},
  the commit timestamps of all $\RSTXN$ transactions are unique.
\end{proof}

\begin{lemma} \label{lemma:visrs-arrs}
  $\visrs \subseteq \arrs.$
\end{lemma}

\begin{proof} \label{proof:visrs-arrs}
  Consider any two transactions $\txnvar, \txnvar' \in \RSTXN$ such that $\txnvar \rel{\visrs} \txnvar'$.
  By Definition~\ref{def:visrs} of $\visrs$, $\txnvar.\committs \le \txnvar'.\readts$.
  Since $\txnvar'.\readts < \txnvar'.\committs$,
  we have $\txnvar.\committs < \txnvar'.\committs$.
  By Definition~\ref{def:arrs} of $\arrs$, $\txnvar \rel{\arrs} \txnvar'$.
\end{proof}

\begin{lemma} \label{lemma:rs-prefixaxiom}
  $\ae \models \prefixaxiom.$
\end{lemma}

\begin{proof} \label{proof:rs-prefixaxiom}
  Consider transactions $\txnvar_{1}, \txnvar_{2}, \txnvar_{3} \in \RSTXN$
  such that $\txnvar_{1} \rel{\arrs} \txnvar_{2} \rel{\visrs} \txnvar_{3}$.
  By Definition~\ref{def:arrs} of $\arrs$ and Definition~\ref{def:visrs} of $\visrs$,
  $\txnvar_{1}.\committs < \txnvar_{2}.\committs \le \txnvar_{3}.\readts$.
  By Definition~\ref{def:visrs} of $\visrs$,
  $\txnvar_{1} \rel{\visrs} \txnvar_{3}$.
\end{proof}

\begin{lemma} \label{lemma:rs-cbaxiom}
  $\ae \models \cbaxiom.$
\end{lemma}

\begin{proof} \label{proof:rs-cbaxiom}
  Consider two transactions $\txnvar, \txnvar' \in \RSTXN$
  such that $\txnvar \rel{\cbrs} \txnvar'$.
  By line~\code{\ref{alg:rs-primary}}{\ref{line:rscommit-tick}},
  $\txnvar.\committs < \txnvar'.\committs$.
  By Definition~\ref{def:arrs} of $\arrs$, $\txnvar \rel{\arrs} \txnvar'$.
\end{proof}

\begin{lemma} \label{lemma:rs-conflict-arrs-tid}
  $\forall \txnvar, \txnvar' \in \RSTXN.\;
    \txnvar \conflict \txnvar' \implies (\txnvar \rel{\arrs} \txnvar'
      \iff \txnvar.\tid < \txnvar'.\tid).$
\end{lemma}

\begin{proof} \label{proof:rs-conflict-arrs-tid}
  Consider any two transactions $\txnvar, \txnvar' \in \RSTXN$
  such that $\txnvar \conflict \txnvar'$.
  By Lemma~\ref{lemma:rs-noconflictaxiom},
  $\txnvar \rel{\visrs} \txnvar' \lor \txnvar' \rel{\visrs} \txnvar$.
  In the following, we proceed in two directions.

  Suppose that $\txnvar \rel{\arrs} \txnvar'$.
  By Lemma~\ref{lemma:visrs-arrs}, $\lnot(\txnvar' \rel{\visrs} \txnvar)$.
  Therefore, $\txnvar \rel{\visrs} \txnvar'$.
  By Lemma~\ref{lemma:visrs-viswt}, $\txnvar \rel{\viswt} \txnvar'$.
  By Lemma~\ref{lemma:viswt-tid}, $\txnvar.\tid < \txnvar'.\tid$.

  Suppose that $\txnvar.\tid < \txnvar'.\tid$.
  By Lemma~\ref{lemma:viswt-tid}, $\lnot(\txnvar' \rel{\viswt} \txnvar)$.
  By Lemma~\ref{lemma:visrs-viswt}, $\lnot(\txnvar' \rel{\visrs} \txnvar)$.
  Therefore, $\txnvar \rel{\visrs} \txnvar'$.
  By Lemma~\ref{lemma:visrs-arrs}, $\txnvar \rel{\arrs} \txnvar'$.
\end{proof}

\begin{lemma} \label{lemma:rs-intaxiom}
  $\ae \models \intaxiom.$
\end{lemma}

\begin{proof} \label{proof:rs-intaxiom}
  Consider a transaction $\txnvar \in \RSTXN$.
  Let $\evar$ be an event such that $\op(\evar) = \readevent(\keyvar, \valvar)$
  is an \emph{internal} $\rsread$ operation in $\txnvar$.
  Let $\evar' \triangleq \max_{\po}\set{\fvar \mid \op(\fvar) = \_(\keyvar, \_) \land \fvar \rel{\po} e}$.
  We need to show that $\evar' = \_(\keyvar, \valvar)$.

  By Lemma~\ref{lemma:rsvis-welldefined}, $\rsvis(e') = \rsvis(e)$.
  If $\op(\evar') = \readevent(\keyvar, \_)$, $\evar$ obtains the same value as $\evar'$.
  Therefore, $\op(\evar') = \readevent(\keyvar, \valvar)$.
  If $\op(\evar') = \writeevent(\keyvar, \_)$, by the visibility rule,
  $\evar'$ is visible when $\evar$ occurs.
  Thus, $\evar$ reads from $\evar'$.
  Therefore, $\op(\evar') = \writeevent(\keyvar, \valvar)$.
\end{proof}

\begin{lemma} \label{lemma:rs-extaxiom}
  $\ae \models \extaxiom.$
\end{lemma}

\begin{proof} \label{proof:rs-extaxiom}
  Consider a transaction $\txnvar \in \RSTXN$.
  Let $\evar$ be an event such that $\op(\evar) = \readevent(\keyvar, \valvar)$
  is an \emph{external} $\rsread$ operation in $\txnvar$.
  Let $W \triangleq \visrs^{-1}(\txnvar) \cap \WriteTx_{\keyvar}$
  be the set of transactions that update key $\keyvar$ and are visible to $\txnvar$.
  If $W = \emptyset$, obviously $e$ obtains the initial value of $\keyvar$.
  Now suppose that $W \neq \emptyset$.
  By Lemma~\ref{lemma:visrs-rsvis}, $W = \rsvis(\txnvar) \cap \WriteTx_{\keyvar}$.
  By Lemma~\ref{lemma:rs-conflict-arrs-tid},
  the $\arrs|_{W}$ order is consistent with the increasing $\tid$ order of the transactions in $W$,
  which is also the list order at line~\code{\ref{alg:wt-db}}{\ref{line:wtread-call-txnvis}}.
  Therefore, $\evar$ reads from $\max_{\arrs} W$.
  Thus, $\max_{\arrs} W \vdash \writeevent(\keyvar, \valvar)$.

  Consider a transaction $\txnvar \in \RSTXN$.
  Let $e = (\_, \readevent(k, v))$
  be an \emph{external} $\rsread$ operation in $\txnvar$.
  Consider the set
  $W = \visrs^{-1}(\txnvar) \cap \set{\txnvar \mid \txnvar \vdash \Write k : \_}$
  of transactions that update key $\kvar$ and are visible to $\txnvar$.
  If $W$ is empty, obviously $e$ obtains the initial value of $\kvar$.
  Now suppose that $W \neq \emptyset$.
  By Lemma~\ref{lemma:visrs-rsvis},
  it is also the set of transactions that pass the check of $\txnvis(\txnvar, \_)$
  for key $\kvar$ (line~\code{\ref{alg:wt-db}}{\ref{line:wtread-call-txnvis}}).
\end{proof}

\begin{theorem} \label{thm:rsalg-rtsi}
  $\rsalg \models \rtsi.$
\end{theorem}

\begin{proof} \label{proof:rsalg-rtsi}
  For any history $\h$ of $\rsalg$,
  we construct an abstract execution $\ae = (\h, \visrs, \arrs)$,
  where $\visrs$ and $\arrs$ are given in Definitions~\ref{def:visrs}
  and \ref{def:arrs}, respectively.
  By Lemmas~\ref{lemma:visrs-acyclic}, \ref{lemma:rs-rbaxiom}, \ref{lemma:rs-noconflictaxiom},
  \ref{lemma:arrs-is-total}, \ref{lemma:visrs-arrs}, \ref{lemma:rs-prefixaxiom}, \ref{lemma:rs-cbaxiom},
  \ref{lemma:rs-intaxiom}, and \ref{lemma:rs-extaxiom}, $\ae \models \rtsi$.
  Since $\h$ is arbitrary, $\rsalg \models \rtsi$.
\end{proof}

Additionally, we illustrate by counterexample that the abstract execution $\ae$
constructed in the proof of Theorem~\ref{thm:rsalg-rtsi} does not satisfy \inrbaxiom.
\begin{example}[$\ae \not\models \realtimesnapshotaxiom$] \label{ex:rs-not-inrbaxiom}
  Consider a transaction $\txnvar \in \RSTXN$ that has committed locally,
  but has not been majority committed.
  Let $\txnvar' \neq \txnvar$ be a transaction that starts before $\txnvar$ finishes.
  That is, $\lnot (\txnvar \rel{\rbrs} \txnvar')$.
  Due to the ``speculative majority'' strategy,
  it is possible for $\txnvar'$ to read from $\txnvar$ so that $\txnvar \rel{\visrs} \txnvar'$.
  Therefore, $\ae \not\models \inrbaxiom$.
\end{example}


\subsection{Correctness of \scalg} \label{ss:sc-correctness}

Consider a history $\h = (\SCTXN, \sosc)$ of \scalg{},
where we restrict $\SCTXN$ to be set of committed transactions.
Denote by $\SCTXNUPDATE \subseteq \SCTXN$ and $\SCTXNRO \subseteq \SCTXN$
be the set of update transactions and read-only transactions in $\SCTXN$, respectively.
We prove that $\h$ satisfies \sessionsi{}
by constructing an abstract execution $\ae = (\h, \vissc, \arsc)$ (Theorem~\ref{thm:scalg-sessionsi}).

In \scalg, all transactions have a read timestamp.
However, only update transactions will be assigned a commit timestamp in 2PC.
For the proof, we need to assign commit timestamps to read-only transactions as well.

\begin{definition}[Commit Timestamps for Read-only \scalg{} Transactions] \label{def:sc-committs}
  Let $\txnvar \in \SCTXNRO$ be a read-only transaction.
  We define $\txnvar.\committs = \txnvar.\readts$.
\end{definition}

We also use Lamport clocks~\cite{Time:CACM1978} of transactions as usual to break ties when necessary.
Let $\txnvar \in \RSTXN$ be a transaction.
We denote by $\lclock(\txnvar)$ its Lamport clock.
\subsubsection{The Visibility Relation} \label{sss:vis-sc}

\begin{definition}[$\vissc$] \label{def:vissc}
  \begin{align*}
    &\forall \txnvar, \txnvar' \in \SCTXN.\;
      \txnvar \rel{\vissc} \txnvar' \iff \\
      &\quad (\txnvar.\committs < \txnvar'.\readts) \;\lor \\
      &\quad (\txnvar.\committs = \txnvar'.\readts \land \lclock(\txnvar) < \lclock(\txnvar')).
  \end{align*}
\end{definition}

\begin{lemma} \label{lemma:vissc-acyclic}
  $\vissc$ is acyclic.
\end{lemma}

\begin{proof} \label{proof:vissc-acyclic}
  By the property that $\forall \txnvar \in \SCTXNUPDATE.\; \txnvar.\readts < \txnvar.\committs$
  and that Lamport clocks are totally ordered.
\end{proof}

\begin{lemma} \label{lemma:sc-session}
  $\ae \models \sessionaxiom.$
\end{lemma}

\begin{proof} \label{proof:sc-session}
  This holds due to the way how nodes and clients maintain their cluster time
  (namely, they distribute their latest cluster time when sending any messages
  and update it when receiving a larger one in incoming messages; see Section~\ref{sss:rs-hlc})
  and the properties of Lamport clocks.
\end{proof}

\begin{definition} \label{def:scvis}
  For a transaction $\txnvar \in \SCTXN$,
  we define $\scvis(\txnvar) \subseteq \SCTXNUPDATE$ to be the set of update transactions in $\SCTXN$
  that are visible to $\txnvar$ according to the visibility rule
  in $\txnvis$ (line~\code{\ref{alg:wt-db}}{\ref{line:procedure-txnvis}}).
\end{definition}

\begin{lemma} \label{lemma:scvis-welldefined}
  For any transaction $\txnvar \in \SCTXN$,
  $\scvis(\txnvar)$ is well-defined.
  That is, $\scvis(\txnvar)$ does not change over the lifecycle of $\txnvar$.
\end{lemma}


\begin{lemma} \label{lemma:vissc-scvis}
  $\forall \txnvar \in \SCTXN.\; \vissc^{-1}(\txnvar) \cap \SCTXNUPDATE = \scvis(\txnvar).$
\end{lemma}


\begin{lemma} \label{lemma:sc-noconflictaxiom}
  $\ae \models \noconflictaxiom.$
\end{lemma}

\begin{proof} \label{proof:sc-noconflictaxiom}
  Consider any two transactions $\txnvar, \txnvar' \in \SCTXN$
  such that $\txnvar \conflict \txnvar'$.
  Suppose that they both update key $\kvar$ on the primary of some shard.
  Suppose by contradiction that
  $\lnot (\txnvar \rel{\vissc} \txnvar' \lor \txnvar' \rel{\vissc} \txnvar)$.
  By Definition~\ref{def:vissc} of $\vissc$, that is
  $(\txnvar.\readts < \txnvar'.\committs) \land (\txnvar'.\readts < \txnvar.\committs)$.

  By Lemma~\ref{lemma:wt-noconflictaxiom},
  $\txnvar$'s $\wtcommit$ finishes before $\txnvar'$'s $\wtstart$ starts
  \emph{or} $\txnvar'$'s $\wtcommit$ finishes before $\txnvar$'s $\wtstart$ starts.
  These two cases are symmetric.
  In the following, we consider the first case
  which implies that $\txnvar$ updates $\keyvar$ before $\txnvar'$ does.
  When $\txnvar'$ updates $\keyvar$,
  $\tuple{\txnvar.\tid, \_, \txnvar.\committs} \in \store[\kvar]$
  (line~\code{\ref{alg:wt-db}}{\ref{line:wtupdate-forall}}).
  However, the check $\txnvis(\txnvar', \_, \txnvar.\committs)$
  (line~\code{\ref{alg:wt-db}}{\ref{line:wtupdate-call-txnvis}}) fails
  because $\txnvar'.\readts < \txnvar.\committs$.
  Therefore, $\txnvar'$ would abort
  (line~\code{\ref{alg:wt-db}}{\ref{line:wtupdate-call-rollback}}).
\end{proof}
\subsubsection{The Arbitrary Relation} \label{sss:sc-ar}

\begin{definition}[$\arsc$] \label{def:arsc}
  \begin{align*}
    &\forall \txnvar, \txnvar' \in \SCTXN.\;
      \txnvar \rel{\arsc} \txnvar' \iff \\
      &\quad (\txnvar.\committs < \txnvar'.\committs) \;\lor \\
      &\quad (\txnvar.\committs = \txnvar'.\committs \land \lclock(\txnvar) < \lclock(\txnvar'))
  \end{align*}
\end{definition}

\begin{lemma} \label{lemma:arsc-total}
  $\arsc$ is a total order.
\end{lemma}

\begin{proof} \label{proof:arsc-total}
  It is easy to show that $\arsc$ is irreflexive, transitive, and total.
\end{proof}

\begin{lemma} \label{lemma:vissc-arsc}
  $\vissc \subseteq \arsc.$
\end{lemma}

\begin{proof} \label{proof:vissc-arsc}
  By a case analysis based on Definition~\ref{def:vissc} of $\vissc$ and Definition~\ref{def:arsc} of $\arsc$.
\end{proof}

\begin{lemma} \label{lemma:sc-prefixaxiom}
  $\ae \models \prefixaxiom.$
\end{lemma}

\begin{proof} \label{proof:sc-prefixaxiom}
  Consider transactions $\txnvar_{1}, \txnvar_{2}, \txnvar_{3} \in \SCTXN$
  such that $\txnvar_{1} \rel{\arsc} \txnvar_{2} \rel{\vissc} \txnvar_{3}$.
  It is easy to show that $\txnvar_{1} \rel{\vissc} \txnvar_{3}$
  by a case analysis based on Definition~\ref{def:arsc} of $\arsc$ and Definition~\ref{def:vissc} of $\vissc$.
\end{proof}

\begin{lemma} \label{lemma:sc-conflict-arsc-tid}
  $\forall \txnvar, \txnvar' \in \SCTXN.\;
    \txnvar \conflict \txnvar' \implies (\txnvar \rel{\arsc} \txnvar'
      \iff \txnvar.\tid < \txnvar'.\tid).$
\end{lemma}

\begin{lemma} \label{lemma:sc-intaxiom}
  $\ae \models \intaxiom.$
\end{lemma}

\begin{proof} \label{proof:sc-intaxiom}
  Consider a transaction $\txnvar \in \SCTXN$.
  Let $\evar = (\_, \readevent(\kvar, \vvar))$
  be an \emph{internal} $\rsread$ operation in $\txnvar$.
  Moreover, $\po^{-1}(\evar) \cap \HEvent_{\evar} \neq \emptyset$.
  Let $\evar' \triangleq \max_{\po}(\po^{-1}(e) \cap \HEvent_{\kvar})$.
  By Lemma~\ref{lemma:scvis-welldefined}, compared with $\evar'$,
  no additional update operations on $\kvar$ (from other transactions)
  are visible to $\evar$ at line~\code{\ref{alg:wt-db}}{\ref{line:wtread-call-txnvis}}.
  Thus, $e' = (\_, \_(\kvar, \vvar))$.
\end{proof}

\begin{lemma} \label{lemma:sc-extaxiom}
  $\ae \models \extaxiom.$
\end{lemma}

\begin{proof} \label{proof:sc-extaxiom}
  Consider a transaction $\txnvar \in \SCTXN$.
  Let $e = (\_, \readevent(k, v))$
  be an \emph{external} $\scread$ operation in $\txnvar$.
  Consider the set
  $W = \vissc^{-1}(\txnvar) \cap \set{\txnvar \mid \txnvar \vdash \Write k : \_}$
  of transactions that update key $\keyvar$ and are visible to $\txnvar$.
  If $W$ is empty, obviously $e$ obtains the initial value of $\keyvar$.
  Now suppose that $W \neq \emptyset$.
  By Lemma~\ref{lemma:vissc-scvis},
  it is also the set of transactions that pass the check of $\txnvis(\txnvar, \_)$
  for key $\keyvar$ (line~\code{\ref{alg:wt-db}}{\ref{line:wtread-call-txnvis}}).

  By Lemma~\ref{lemma:sc-conflict-arsc-tid},
  the $\arsc|_{W}$ order is consistent with the increasing $\tid$ order of the transactions in $W$,
  which is the list order at line~\code{\ref{alg:wt-db}}{\ref{line:wtread-call-txnvis}}.
  Therefore, $\max_{\arsc} W \vdash \Write \kvar : \vvar$.
\end{proof}

\begin{theorem} \label{thm:scalg-sessionsi}
  $\scalg \models \sessionsi.$
\end{theorem}

\begin{proof} \label{proof:scalg-sessionsi}
  For any history $\h$ of $\scalg$,
  we construct an abstract execution $\ae = (\h, \vissc, \arsc)$,
  where $\vissc$ and $\arsc$ are given in Definitions~\ref{def:vissc}
  and \ref{def:arsc}, respectively.
  By Lemmas~\ref{lemma:vissc-acyclic}, \ref{lemma:sc-session}, \ref{lemma:sc-noconflictaxiom},
  \ref{lemma:arsc-total}, \ref{lemma:vissc-arsc}, \ref{lemma:sc-prefixaxiom},
  \ref{lemma:sc-intaxiom}, and \ref{lemma:sc-extaxiom},
  $\ae \models \sessionsi$.
  Since $\h$ is arbitrary, $\scalg \models \sessionsi$.
\end{proof}


\end{document}